\documentclass[11pt]{article}
\usepackage[T1]{fontenc}
\usepackage[utf8]{inputenc}
\usepackage{lmodern}
\usepackage[margin=1.05in]{geometry}
\usepackage{amsmath,amssymb,amsthm,mathtools}
\usepackage{microtype,booktabs,tabularx,array}
\usepackage{aliascnt}
\usepackage{enumitem}
\usepackage{xcolor}
\usepackage[unicode,colorlinks=true,linkcolor=blue!45!black,citecolor=blue!45!black,urlcolor=blue!45!black]{hyperref}
\usepackage[nameinlink,noabbrev,capitalise]{cleveref}
\setlist{itemsep=2pt,topsep=4pt}
\numberwithin{equation}{section}
\newtheorem{theorem}{Theorem}[section]
\newtheorem*{inversionrestated}{Lemma~\ref{lem:inversion}}
\newtheorem*{operatorboundrestated}{Corollary~\ref{cor:main-envelope}}
\newtheorem*{stateboundrestated}{Corollary~\ref{cor:state-envelope}}
\newaliascnt{lemma}{theorem}
\newtheorem{lemma}[lemma]{Lemma}
\aliascntresetthe{lemma}
\crefname{lemma}{Lemma}{Lemmas}
\Crefname{lemma}{Lemma}{Lemmas}
\newaliascnt{proposition}{theorem}
\newtheorem{proposition}[proposition]{Proposition}
\aliascntresetthe{proposition}
\crefname{proposition}{Proposition}{Propositions}
\Crefname{proposition}{Proposition}{Propositions}
\newaliascnt{corollary}{theorem}
\newtheorem{corollary}[corollary]{Corollary}
\aliascntresetthe{corollary}
\crefname{corollary}{Corollary}{Corollaries}
\Crefname{corollary}{Corollary}{Corollaries}
\theoremstyle{definition}
\newaliascnt{definition}{theorem}
\newtheorem{definition}[definition]{Definition}
\aliascntresetthe{definition}
\crefname{definition}{Definition}{Definitions}
\Crefname{definition}{Definition}{Definitions}
\newaliascnt{assumption}{theorem}

\aliascntresetthe{assumption}
\crefname{assumption}{Assumption}{Assumptions}
\Crefname{assumption}{Assumption}{Assumptions}
\theoremstyle{remark}
\newaliascnt{remark}{theorem}
\newtheorem{remark}[remark]{Remark}
\aliascntresetthe{remark}
\crefname{remark}{Remark}{Remarks}
\Crefname{remark}{Remark}{Remarks}
\newcommand{\C}{\mathbb C}
\newcommand{\F}{\mathbb F}
\newcommand{\E}{\mathbb E}
\newcommand{\HH}{\mathcal H}
\newcommand{\Stab}{\operatorname{Stab}}
\newcommand{\Cliff}{\operatorname{Cliff}}
\newcommand{\Pauli}{\mathcal P}
\newcommand{\PG}{\overline{\mathcal P}}
\newcommand{\Zgrp}{\mathcal Z}
\newcommand{\LC}{\Lambda_{\mathrm C}}
\newcommand{\LS}{\Lambda_{\mathrm S}}
\newcommand{\LM}{\Lambda_{\mathrm M}}
\newcommand{\rC}{\chi_{\mathrm C}}
\newcommand{\FS}{F_{\mathrm S}}
\newcommand{\Bnd}{\mathcal B}
\newcommand{\Snd}{\mathcal S}
\newcommand{\QAND}{Q}
\newcommand{\tr}{\operatorname{Tr}}
\newcommand{\rank}{\operatorname{rank}}
\newcommand{\sr}{\operatorname{sr}}
\newcommand{\SR}{\operatorname{SR}}
\newcommand{\supp}{\operatorname{supp}}
\newcommand{\Span}{\operatorname{span}}
\newcommand{\poly}{\operatorname{poly}}
\newcommand{\negl}{\operatorname{negl}}
\newcommand{\AND}{\operatorname{AND}}
\newcommand{\NOR}{\operatorname{NOR}}

\newcommand{\eps}{\varepsilon}
\newcommand{\ket}[1]{\lvert #1\rangle}
\newcommand{\bra}[1]{\langle #1\rvert}
\newcommand{\ip}[2]{\langle #1,#2\rangle}
\newcommand{\norm}[1]{\lVert #1\rVert}
\newcommand{\floor}[1]{\left\lfloor #1\right\rfloor}
\newcommand{\ceil}[1]{\left\lceil #1\right\rceil}
\newcommand{\J}[1]{\lvert #1\rangle\!\rangle}
\newcommand{\Dtr}{D_{\mathrm{tr}}}
\DeclareMathOperator{\diag}{diag}
\hypersetup{pdftitle={Improved bounds on stabilizer extent and Clifford rank},pdfauthor={Pulkit Sinha and Benjamin Lovitz}}
\title{Improved bounds on stabilizer extent and Clifford rank}
\author{%
Pulkit Sinha\textsuperscript{1,2} \qquad Benjamin Lovitz\textsuperscript{3}\\[0.6em]
\small\textsuperscript{1}David R. Cheriton School of Computer Science, University of Waterloo\\
\small\textsuperscript{2}Institute for Quantum Computing, University of Waterloo\\
\small\textsuperscript{3}Concordia University
}

\begin{document}
\maketitle
\begin{abstract}
We prove that every pure state of stabilizer rank at most $k$ has stabilizer
extent at most $2^{O(\sqrt{k\log(k+1)})}$, and establish the analogous
bound for the squared Clifford coefficient norm of Clifford-rank-$k$ operators.
This implies stabilizer fidelity at least $2^{-O(\sqrt{k\log(k+1)})}$, resolving the quantitative conjecture of \cite{MT25}, and proves an $\Omega(n^2/\log n)$ lower bound for the approximate stabilizer rank of tensor powers of any non-stabilizer qubit state. The latter result generalizes the best-known lower bound for tensor powers of $T$-states \cite{MT24} to arbitrary non-stabilizer qubit states, including magic states.

As a consequence of the Clifford rank--norm inequality, we obtain an $\Omega(n^2/\log n)$ lower bound for exact representations of $n$-bit AND function by quadratic phases, improving the
previous best-known linear bound. Further consequences rule out pseudorandom state
and unitary ensembles with approximate stabilizer and Clifford rank $O((\log n)^2/\log\log n)$, respectively, a $\log n$ improvement over prior work \cite{KS26}. We also obtain tomography algorithms for states of stabilizer rank at most $k$, with $\poly(n)2^{O(\sqrt{k\log(k+1)})}$ time and copy complexity, a nearly square-root improvement in the exponent over the best-known algorithm.
\end{abstract}

% \clearpage
% \setcounter{tocdepth}{1}
% \tableofcontents
% \clearpage

\section{Introduction}\label{sec:intro}

The Gottesman--Knill theorem identifies a large class of quantum
computations that can be simulated efficiently on a classical computer:
Clifford circuits acting on stabilizer states, with Pauli measurements
\cite{Got97,Got98,AG04}. The reason is that a stabilizer state has a short
classical description that remains easy to update throughout the
computation. This remains true even when the state is highly entangled.
For example, a Bell pair can be specified by its two Pauli stabilizers,
$X\otimes X$ and $Z\otimes Z$, rather than by listing its amplitudes.
More generally, $n$ independent commuting Pauli constraints specify an
$n$-qubit stabilizer state. The description grows polynomially with $n$,
even though the vector has $2^n$ amplitudes.
To understand the cost of simulating more general quantum computations,
it is therefore natural to ask how far this description can be extended.

One approach is to write the input as a linear combination of stabilizer
states and follow the evolution of each term. Clifford gates preserve
the stabilizer property, and Pauli measurement outcomes send each term
to another stabilizer state, up to a scalar, or to zero. The number of
terms thus measures the size of the description that the simulator must
maintain. This idea underlies simulation methods based on stabilizer
rank \cite{BSS16,BG16}. It is particularly useful when the non-Clifford
resources of a computation are supplied as magic states: the rest of
the computation can then be handled within the stabilizer formalism.

The usefulness of a decomposition depends on the freedom to choose its
terms. Expanding in the computational basis always gives stabilizer
states, but it can miss the structure that makes simulation efficient.
For instance, $\ket+^{\otimes n}$ has nonzero amplitude on every
computational-basis vector, yet is itself a single stabilizer state.
Allowing all stabilizer states as building blocks can therefore compress
a description enormously. Finding the best such representation is a
mathematical problem in its own right.

\subsection{Stabilizer rank and extent}
For a normalized state $\ket\psi$, its \emph{stabilizer rank}
$\chi(\ket\psi)$ is the smallest number of terms in a decomposition
\[
 \ket\psi=\sum_{j=1}^k c_j\ket{s_j},
\]
where the $c_j$ are complex coefficients and the $\ket{s_j}$ are normalized stabilizer states.
% For tensor powers of a magic state,
% this allows a decomposition to exploit structure shared across different
% copies. A rank lower bound must rule out every such possibility.

A second measure, the \emph{stabilizer extent}, asks for small coefficients:
\[
 \xi(\ket\psi)=\min_{\ket\psi=\sum_j c_j\ket{s_j}}
                    \left(\sum_j|c_j|\right)^2.
\]
% The minimizations need not select the same decomposition. Rank favors
% few terms, while extent allows additional terms if they reduce the total
% coefficient size. 
Extent controls sparsification: a decomposition with
small coefficient sum can be sampled to produce a short approximate
decomposition \cite{BBC19}. Thus both quantities describe useful ways
of extending stabilizer simulation beyond the Gottesman--Knill setting.
More quantitatively, if $\chi_\delta(\ket\psi)$ denotes the minimum
stabilizer rank of a vector within Euclidean distance $\delta$ of
$\ket\psi$, sparsification gives
\[
 \chi_\delta(\ket\psi)\le1+\frac{\xi(\ket\psi)}{\delta^2}
 \qquad (\delta>0)
\]
\cite{BBC19}. A bound on extent therefore gives an upper bound on the
number of terms needed for approximation, even when the decomposition
attaining the extent uses many terms. The converse direction is the
one relevant to lower bounds: can a state with large extent nevertheless
have a short exact stabilizer decomposition?

Extent is closely related to the largest overlap with a stabilizer state.
For a normalized state $\ket\psi$, define its \emph{stabilizer fidelity} by
\[
 \FS(\ket\psi)=\max_{\ket s\text{ stabilizer}}|\ip{s}{\psi}|^2.
\]
The triangle inequality applied to a stabilizer decomposition gives
$\xi(\ket\psi)\ge1/\FS(\ket\psi)$ \cite{BBC19}. Thus a state with
small overlap with every stabilizer state must have large extent.

Extent is also known exactly for important families of states.
It is multiplicative for tensor
products whose individual factors have at most three qubits
\cite{BBC19}. For the familiar magic state
$\ket T=(\ket0+e^{i\pi/4}\ket1)/\sqrt2$, these results give
\[
 \xi(\ket T^{\otimes n})=\cos^{-2n}(\pi/8).
\]
%Such formulas give a precise measure of the coefficient cost for
%important families of states.
Multiplicativity means that allowing entangled stabilizer states across the copies cannot improve the optimal coefficient sum obtained from single-copy decompositions. We note that multiplicativity can fail for general states \cite{HMMVG21}. 

Even though the extent of $\ket T^{\otimes n}$ is known exactly, the above
formula does not determine how many stabilizer terms are needed to
represent it. The corresponding rank questions are much harder because
of cancellation. For an
orthonormal expansion of a unit vector, the squared coefficients sum
to one, and Cauchy--Schwarz bounds their absolute sum by $\sqrt k$.
Stabilizer states in a decomposition need not be orthogonal. Large
coefficients may cancel, leaving a normalized state even though the
coefficient sum is large. A lower bound on extent therefore becomes a
rank lower bound only if we can control how much this cancellation helps.

% It is also important to distinguish a property of a particular
% representation from a property of the state. A supplied decomposition
% may have a large coefficient sum even if another decomposition of the
% same state has a small one. To bound extent, it is enough to exhibit
% one good decomposition. We can change the stabilizer states, use more
% terms, and take advantage of structure that was not visible in the
% original representation. This distinction makes it possible to seek
% stronger extent bounds than estimates on the original coefficients.

This leads to the question studied here: how strongly does a bound on
stabilizer rank constrain extent?
% independently of the number of qubits?
% The independence from the number of qubits is important. Expanding any
% normalized state in the computational basis gives the elementary bound
% $\xi(\ket\psi)\le2^n$, but this says little about a state represented by
% only a few stabilizer terms. We seek a bound that remains useful as the
% ambient Hilbert space grows, provided the number of terms stays small.
% Such a result would let us turn the known extent of a family of states
% into a lower bound on its rank, without analyzing all possible
% decompositions of each state separately.
Mehraban--Tahmasbi \cite{MT25} showed that bounded rank forces positive
stabilizer fidelity independent of the ambient dimension, and conjectured
that $\FS(\ket\psi)\ge2^{-k}/C$ for normalized states of stabilizer rank
at most $k$, where $C>0$ is a universal constant. Kalra--Sinha \cite{KS26}
used Barnes--Wall lattices to control the coefficients in any supplied
linearly independent stabilizer decomposition. Their result implies
$\xi(\ket\psi)\le2^{O(k\log(k+1))}$ when $\chi(\ket\psi)\le k$.
We obtain the stronger bound $\xi(\ket\psi)\le2^{O(\sqrt{k\log(k+1)})}$
%by allowing the decomposition itself to change.
by optimizing the decomposition.
This implies stabilizer
fidelity at least $2^{-O(\sqrt{k\log(k+1)})}$, resolving the quantitative
conjecture of Mehraban--Tahmasbi \cite{MT25}, and moreover, improving their conjectured scaling by a square-root factor in the exponent.

\subsection{Main results}
\begin{theorem}[Stabilizer rank--extent bound]\label{thm:intro-state}
There is a universal constant $c>0$ such that, for every $n\ge1$, every
$k\ge1$, and every vector $\ket v$ on $n$ qubits with $\chi(\ket v)\le k$,
\begin{equation}\label{eq:intro-state-main}
 \xi(\ket v)\le2^{c\sqrt{k\log_2(k+1)}}\norm{\ket v}^2.
\end{equation}
\end{theorem}
In particular, every unit vector in the span of $k$
stabilizer states has a decomposition of small total coefficient size.
The stabilizer states in this new decomposition may differ from those
spanning the original space. This freedom is essential to our proof:
we replace groups of terms by better decompositions inside stabilizer codes.

The same question can be asked for operators. Let $\rC(A)$ be the
minimum number of Clifford unitaries whose linear combination equals
$A$, and define the Clifford coefficient norm by
\[
 \LC(A)=\min_{A=\sum_j c_j C_j}\sum_j|c_j|,
 \qquad C_j\text{ Clifford unitaries}.
\]
Here we retain the unsquared norm, so $\LC(A)^2$ is the analogue of
state extent. We write $\norm A_F=\sqrt{\tr(A^\dagger A)}$ for the
Frobenius norm.
\begin{theorem}[Clifford rank--coefficient norm bound]\label{thm:intro-main}
There is a universal constant $c>0$ such that, for every $n\ge1$, every
$k\ge1$, and every operator $A$ on $n$ qubits with $\rC(A)\le k$,
\begin{equation}\label{eq:intro-main}
 \LC(A)^2\le2^{c\sqrt{k\log_2(k+1)}}2^{-n}\norm A_F^2.
\end{equation}
\end{theorem}
The factor $2^{-n}\norm A_F^2$ is the squared norm of the image of $A$
under the normalized Choi map; it equals one when $A$ is unitary. The theorem also
applies to nonunitary operators, which will be important for the application to Boolean
functions below. We prove explicit versions of these results that do not rely on unknown constants in
\cref{thm:state-finite-main,thm:finite-main}.

\subsection{Consequences}
\paragraph{Fidelity and magic-state rank.}
A small coefficient norm forces a state to have appreciable overlap
with at least one stabilizer state. Indeed, taking the inner product
of a normalized $\ket\psi$ with its own decomposition gives
$1\le\sqrt{\FS(\ket\psi)}\sum_j|c_j|$. \cref{thm:intro-state} therefore implies
\[
 \chi(\ket\psi)\le k
 \quad\Longrightarrow\quad
 \FS(\ket\psi)\ge2^{-O(\sqrt{k\log(k+1)})}.
\]
More generally,
if $V$ is spanned by $k$ stabilizer states, then
\[
 \norm{P_V\ket\phi}^2
 \le2^{O(\sqrt{k\log(k+1)})}\FS(\ket\phi).
\]
The latter result controls the approximate stabilizer rank, as a state
close to $V$ must have a large projection onto $V$. For every fixed non-stabilizer one-qubit pure state $\ket\phi$, stabilizer
fidelity of $\ket\phi^{\otimes n}$ decays exponentially in $n$
\cite{BBC19}. Putting these together, we obtain the following:
\begin{corollary}
Let $\ket \phi \in \mathbb{C}^2$ be any non-stabilizer pure state. Then
\[
 \chi_\delta(\ket\phi^{\otimes n})
 =\Omega_{\phi,\delta}(n^2/\log n)
 \qquad\text{for every fixed }0\le\delta<1.
\]
\end{corollary}
This result is proven in \cref{cor:magic-rank}. Lower bounds on stabilizer rank have been developed through several
approaches \cite{PSV22,LS22,Labib22}. In particular,
Mehraban--Tahmasbi \cite{MT24} already proved a nearly quadratic
constant-error lower bound for $T$-state tensor powers. Our result generalizes this to arbitrary one-qubit non-stabilizer states. We remark that our proof technique is quite different from that of~\cite{MT24}, which was based on Clifford+T unitary synthesis.
% It also gives nearly quadratic rank
% when the squared target overlap is exponentially small, provided it is
% exponentially larger than the stabilizer fidelity
% (\cref{cor:tiny-overlap}).

\paragraph{Quadratic phases and AND.}
\cref{thm:intro-main} has a consequence for a classical representation
problem. How many quadratic phases $(-1)^{q(x)}$, with $q$ a Boolean
polynomial of degree at most two, are needed to express $\AND_n$ as a
complex linear combination? Write $R_2(f)$ for the minimum number
needed to express a function $f$, and $\QAND(n)=R_2(\AND_n)$.
We prove the following:
\begin{corollary}
It holds that
\begin{equation}\label{eq:intro-and}
 \QAND(n)=\Omega(n^2/\log n).
\end{equation}
\end{corollary}
This improves Lovett's linear lower bound, presented in \cite[Appendix~A]{Williams18}.
% The connection is that each quadratic phase is the diagonal of a
% Clifford unitary.
The connection is that the diagonal matrix of a quadratic phase is a
Clifford unitary.Thus a short quadratic-phase representation of $f$
gives a short Clifford representation of $\diag(f(x))$. Every Clifford
unitary has operator norm one, so the coefficient norm must be at least
$\max_x|f(x)|$. \cref{thm:intro-main} then prevents a nonzero sum of $k$ phases
from being too concentrated: its support has size at least
$2^{n-O(\sqrt{k\log(k+1)})}$. Since AND is supported at a single input,
\eqref{eq:intro-and} follows.
% We emphasize that this result concerns exact representation: the zero function already approximates AND with normalized squared $L_2$ error $2^{-n}$. This uses the operator associated with the
% function; the normalized amplitude vector of AND is itself a stabilizer
% state. 

\paragraph{Pseudorandomness.} The fidelity bound also makes low rank detectable. If
$k=O((\log n)^2/\log\log n)$, a state within any fixed Euclidean distance
below one of a stabilizer-rank-$k$ vector has inverse-polynomial stabilizer fidelity.
The Bell-difference test of \cite{GIKL23} then distinguishes it from Haar
random using polynomially many copies. Consequently, a pseudorandom state
ensemble can contain such states only with negligible probability
\begin{corollary}\label{thm:prs-intro}
For every fixed $C>0$ and $0\le\delta<1$, a pseudorandom pure-state ensemble $\{\ket{\psi_{n,s}}\}_s$ on $n$ qubits must satisfy
\begin{equation*}
 \Pr_s\left[
 \chi_\delta(\ket{\psi_{n,s}})
 \le C\frac{(\log_2n)^2}{\log_2\log_2n}
 \right]=\negl(n).
\end{equation*}
\end{corollary}
%$\negl(n)$ The expression is used for sufficiently large $n$.
See \cref{thm:prs}. This raises the rank threshold from $O(\log n/\log\log n)$
obtained in \cite{KS26}. An analogous bound on the approximate Clifford rank of pseudorandom unitaries can be proven as a consequence of \cref{thm:intro-main}, see~\cref{thm:pru}.
% Testing a random column of a unitary gives the
% same exclusion for approximate Clifford rank, with Frobenius error
% at most $\delta\sqrt{2^n}$ for any fixed $\delta<1$.

\paragraph{Learning.}
Substituting~\cref{thm:intro-state} into the bounded-extent results
of \cite{ADGG26} gives the following upper bound on the time and sample complexity of low stabilizer rank tomography.
% gives state tomography under an exact stabilizer-rank-$k$ promise
% with copy and time complexity
% \begin{equation}\label{eq:intro-tomo}
%  \poly(n,1/\eps,\log(1/\beta))\,
%  2^{O(\sqrt{k\log(k+1)})},
% \end{equation}
% where $\eps$ is trace-distance error and $\beta$ is failure probability.

\begin{corollary}\label{thm:state-tomography-intro}
Given copies of an unknown normalized $n$-qubit pure state $\ket{\psi}$ with
$\chi(\ket{\psi})\le k$, a classical description of a pure state within trace
distance $\eps$ can be produced with failure probability at most $\beta$
with sample and time complexity
\begin{equation*}
 \poly(n,1/\eps,\log(1/\beta))\,
 2^{O(\sqrt{k\log(2k)})}.
\end{equation*}
\end{corollary}
This improves on the previous $k^k$ time complexity obtainable from the stabilizer extent bound proven in~\cite{KS26}. See~\cref{thm:state-tomography}. We note that the unitary and
Hamiltonian results in \cite[Section~8.2.5]{ADGG26}, using the framework
of \cite{DJJMP26}, can be used to prove upper bounds on the time and copy complexity of learning unitaries and Hamiltonians of bounded Clifford rank, see~\cref{sec:clifford-tomography,sec:hamilonian-tomography}.
% give the same dependence on rank under their respective
% access models; 

% Hamiltonian learning also depends polynomially on a known
% operator-norm bound.

\subsection{Main techniques}

The key idea is randomized projection onto stabilizer codespaces. Starting
from a sum of $k$ stabilizer states,
\[\ket \psi=\sum_{i=1}c_i\ket{s_i}\]
we seek stabilizer codespace projections that remove many
terms at once while leaving the others as stabilizer states. That is, we aim to find stabilizer codespaces $\mathcal C$ in the ambient space such that a large number of $\ket{s_i}$ are orthogonal to $\mathcal C$. 

At a high level, we do this because it becomes easier to track the extent of the two components of $\ket \psi$: $P_{\mathcal C} \ket \psi$ and $(I-P_{\mathcal C})\ket{\psi}$. It is clear from the choice of $\mathcal C$ that $P_{\mathcal C} \ket{\psi}$ has stabilizer rank smaller than $k$, as many stabilizer terms inside get annihilated, and others get projected to stabilizer states. Dealing with the other part is non-trivial and we discuss our approach in the coming paragraphs.

%If we can
%also reconstruct the original vector $\ket{\psi}$ from these simpler projected vectors $P_{\mathcal C}\ket {\psi}$, where $P_{\mathcal C}$ stands for the orthogonal projector on to $\mathcal C$, as small coefficient linear combinations, say something like

%\[\ket \psi= \beta \ket {\psi'}+\sum \alpha_{\mathcal C} P_{\mathcal \mathcal C} \ket {\psi} \], all coefficients $\beta, \alpha_{\mathcal C}$ being small,
 %we obtain a bound on extent by via recursion through $\ket{\psi'}$ and $\mathcal {\ket{\psi$

%The choice of projections is guided by a code containing a group of terms
%from the decomposition.
% \paragraph{Reducing rank by randomized projection.}
Suppose a low dimension (say $d$) stabilizer codespace $\mathcal C_0$ contains many terms of the
original decomposition, say from $\ket{s_1}$ to $\ket {s_r}$. The component $P_{C_0}\ket{\psi}$ clearly lies in the codespace $\mathcal C_0$. Since any stabilizer codespace admits an orthonormal basis of stabilizer states, we can rewrite $P_{C_0} \ket{\psi}$ in that basis, which will give us a small coefficient stabilizer decomposition for it. In particular, we get that $\xi(P_{\mathcal C_0} \ket{\psi})\leq d \norm{P_{\mathcal C_0} \ket{\psi}}^2$.

To handle the component outside, we would like to remove all the selected terms
while preserving the stabilizer form of the others. We achieve this by choosing a random Pauli constraint satisfied by the
codespace and projecting onto the opposite eigenspace. That is, for any Pauli operator $Q$ such that $\mathcal C_0$ lies inside the the $+1$ eigenspace of $Q$, we look at the the vector $(I-Q) \ket{\psi}$. This removes every selected stabilizer term $\ket {s_1}\dots \ket {s_r}$, while each remaining term $(I-Q)\ket{s_j}$ for $j>r$ is still a stabilizer state up to a scalar, or becomes zero.
Averaging these
unnormalized projected vectors recovers the component outside the
codespace, up to a fixed factor, i.e, $\mathbb E[(I-Q)\ket{\psi}]\propto (I-P_{\mathcal C_0})\ket {\psi}$, where the average is over the choice of $Q$. Their average squared norm is also
controlled by the norm of that component, which keeps the coefficient
cost of this reconstruction small. We can therefore apply induction to
the shorter decompositions and combine them with the component inside
the code to bound the extent of the original vector.

\paragraph{Finding a useful codespace.}
The remaining question is how to find a small codespace containing many
terms of the decomposition. We show that two stabilizer codespaces $\mathcal C_1,\mathcal C_2$ with
large overlap can be merged into a larger stabilizer codespace $\mathcal C\supset C_1,C_2$, whose dimension $\dim(\mathcal C)$ is controlled by their individual dimensions $\dim(C_i)$ and overlap. Here, by overlap we mean $\norm{P_{\mathcal C_1}P_{\mathcal C_2}}_{\infty}$, which also equals the maximum possible inner product between vectors in $\mathcal C_1$ and $\mathcal C_2$. Starting with the
individual stabilizer states $\ket{s_i}$ in the decomposition of $\ket{\psi}$, which form a 1 dimensional stabilizer codespace,
we apply this observation inductively to
group the terms into clusters. Either a low-dimensional codespace
contains many of the terms, or the terms split into clusters whose
containing codespaces are nearly orthogonal. In the first case, the codespace supplies the randomized projections
that reduce the number of terms. In the second case, near orthogonality
limits cancellation between clusters, and a stabilizer basis within
each cluster's codespace gives a decomposition with controlled
coefficient norm. Thus clustering identifies either a useful projection
or a decomposition whose extent we can already bound.

\paragraph{From states to Clifford operators.}
Under the normalized Choi map, Clifford unitaries correspond to
maximally entangled stabilizer states across an $n:n$ bipartition.
An arbitrary stabilizer decomposition of a Choi vector can contain
states that are not maximally entangled, so our state extent theorem
alone does not give the Clifford coefficient-norm bound. Dutt et al.
\cite{DJJMP26} obtain their rank-to-norm estimate directly from the
coefficient bound of \cite{KS26}, which controls the coefficients in the
original Clifford Choi decomposition. Our proof must also control the
new states it introduces.

We do this by tracking the \emph{entanglement deficit}, the number of
ebits missing from maximal entanglement. A stabilizer state with deficit
$a$ can be expanded into maximally entangled stabilizer states with
coefficient sum at most $2^{a/2}$. This follows from the bipartite
stabilizer normal form \cite{FCYBC04}: replace each unentangled pair using
$\ket{00}=(\ket{\Phi^+}+\ket{\Phi^-})/\sqrt2$.
A code containing a stabilizer state of deficit $a$ also has a stabilizer
basis with that same deficit, obtained by Pauli translations.

The crucial control under projection is supplied by
\cite[Lemma~3.4]{DJJMP26}: a Pauli measurement increases the deficit
by at most one. We apply this fact at each elimination step and induct
simultaneously over all deficits. The extra cost changes the recurrence
from $D+2b$ to $D+4b$, preserving the square-root exponent. Thus the
clustering argument extends to Clifford decompositions even though the
intermediate projections leave the maximally entangled class.
An independent reduction in \cref{sec:app-reduction} goes in the other
direction: any operator rank--norm bound yields a state rank--extent
bound with at most $2k-1$ Clifford terms, by choosing a common Clifford
basis and lifting the state amplitudes to a diagonal operator.

\subsection{Organization}
\Cref{sec:state-main,sec:operator} prove the stabilizer and Clifford
rank--norm inequalities. \Crefrange{sec:uncertainty}{sec:tomography}
develop their consequences for quadratic-phase representations,
approximate stabilizer rank, pseudorandomness, and learning. The appendices
contain the state-to-operator reduction and supporting technical details.

\paragraph{Note added:} While preparing this manuscript, we became aware of the simultaneous work~\cite{Dutt2026}, which proves the upper bound
$
 \xi(\ket v)\le2^{O(k)}\norm{\ket v}^2
$
for any vector of stabilizer rank at most $k$. Our~\cref{thm:intro-state} is stronger than this result by a square root factor in the exponent, up to a log factor. As a result, the time complexity of our learning algorithm in~\cref{thm:state-tomography-intro} is also stronger by an almost-square-root factor in the exponent.

\section*{Acknowledgments}
PS acknowledges the support of the Natural Sciences and Engineering Research Council of Canada (NSERC)(RGPIN-2023-03731, ALLRP-578455-2022). Nous remercions
le Conseil de recherches en sciences naturelles et en génie du Canada (CRSNG) de son soutien
(RGPIN-2023-03731,ALLRP-578455-2022). PS is also supported by a Cheriton Graduate Scholarship. Research at IQC is supported by the Government of Canada through
Innovation, Science and Economic Development Canada, and by the Province of Ontario through the Ministry of Research, Innovation and Science. BL acknowledges support from NSERC Discovery Grant number RGPIN-2026-05413.
\paragraph{AI disclosure.}

GPT-6 Astra and GPT 5.6-Sol were used extensively in this work, including to generate mathematical arguments, and to draft and revise the manuscript. The authors have verified all the generated claims and bear the full responsibility for the correctness of both the technical arguments as well as attributions.

\section{Background and conventions}\label{sec:background}
All vector spaces are finite-dimensional. Inner products are conjugate-linear
in the first argument. Logarithms are base two unless written $\ln$.
Constants implicit in $O(\cdot)$ and $\Omega(\cdot)$ are universal unless a
subscript specifies otherwise. We use $\HH_n=(\C^2)^{\otimes n}$ and
$N=2^n$. States called normalized are unit vectors. A vector appearing in
a decomposition or approximation need not be normalized.

\subsection{Pauli groups, stabilizers, and codes}
\begin{definition}[Phase-free Pauli group]\label{def:phase-free-pauli}
Let $\Pauli_n$ be the $n$-qubit Pauli group, including the scalar phases
$\{\pm1,\pm i\}$, and let
\[
 \PG_n=\Pauli_n/\{\pm I,\pm iI\}
\]
be its phase-free quotient. We write $\overline P$ for the class of $P$.
Each class $p$ has a unique representative in $\{I,X,Y,Z\}^{\otimes n}$,
denoted $W_p$.
The group $\PG_n$ has $4^n$ elements, its multiplication is abelian, and
every element has square equal to the identity $I^{\otimes n}$.  The order of any subgroup of $\PG_n$ is a power of two; its number
of independent generators is the base-two logarithm of its order. This group can also be seen as isomorphic to the group $\mathbb Z_2^{2n}$.
We denote the identity class by $e=\overline{I^{\otimes n}}$.

\end{definition}

\begin{lemma}[Pauli trace and transpose identities]\label{lem:pauli-identities}
\[
 \tr(W_pW_q)=\begin{cases}N,&p=q,\\0,&p\ne q,\end{cases}\qquad
 W_p^T=\tau_pW_p,\quad \tau_p\in\{\pm1\}.
\]
\end{lemma}

Throughout, a \emph{commuting} subgroup of $\PG_n$ means one whose
Pauli representatives commute with each other.

\begin{definition}[Clifford unitary]\label{def:clifford-unitary}
A Clifford unitary normalizes $\Pauli_n$.
\end{definition}

\begin{lemma}[Clifford conjugation]\label{lem:clifford-conjugation}
We use two standard facts about
Clifford conjugation: it acts transitively on the nonidentity elements of
$\PG_n$, and it sends any independent commuting collection to distinct
single-qubit $Z$ operators. For a compatible signed stabilizer group, the
signs can also be made positive. These facts and their circuit realization
are part of the standard stabilizer formalism
\cite{Got97,Got98,DDM03,AG04}.
\end{lemma}

\begin{definition}[Stabilizer code]\label{def:stabilizer-code}
A stabilizer code $\mathcal C\subseteq\HH_n$ is the common $+1$
eigenspace of an abelian group $H_{\mathcal C}$ of signed Hermitian Paulis
not containing $-I$. Its canonical stabilizer group consists of all signed
Paulis acting as the identity on the code. Its projector and dimension are
\begin{equation}\label{eq:code-projector}
 P_{\mathcal C}=\frac1{|H_{\mathcal C}|}\sum_{g\in H_{\mathcal C}}g,
 \qquad
 \dim\mathcal C=\frac{2^n}{|H_{\mathcal C}|}.
\end{equation}
A one-dimensional code is a stabilizer ray. We write $\Stab_n$ for the set consisting of one
normalized representative of each such ray.
\end{definition}

\begin{definition}[Signed and phase-free stabilizer groups]\label{def:stabilizer-groups}
For a normalized stabilizer state $\ket s$, distinguish the following groups.
\begin{enumerate}[label=(\roman*)]
\item The \emph{signed stabilizer group} consists of the Pauli operators
that fix $\ket s$, with their signs retained:
\[
 H_s=\{P\in\Pauli_n:P\ket s=\ket s\}.
\]

\item The \emph{phase-free stabilizer group} is obtained by forgetting
the scalar phases of the elements of $H_s$:
\[
 G_s=\{\overline P:P\in H_s\}\le\PG_n.
\]
Here $\overline P$ is the phase-free class of $P$, as in
\cref{def:phase-free-pauli}.
\end{enumerate}
\end{definition}

\begin{lemma}[Maximal commuting Pauli groups]\label{lem:maximal-commuting}
The quotient map is injective on $H_s$, and
$|H_s|=|G_s|=2^n$. The group $G_s$ is maximal commuting:
any Pauli commuting with all its representatives has its class in $G_s$.
Conversely, any maximal commuting subgroup of
$\PG_n$ has order $2^n$ and is Clifford-conjugate to
\[
 \Zgrp_n=\{\overline{Z^x}:x\in\F_2^n\}
          =\langle\overline Z_1,\ldots,\overline Z_n\rangle.
\]
\end{lemma}

\begin{remark}[Stabilizer signs]\label{rem:stabilizer-signs}
The phase-free group alone does not determine the state: the signs in
$H_s$ distinguish its joint eigenvalue assignments. For example,
$\ket0$ and $\ket1$ have the same $G_s$ but different $H_s$.
We use signed groups for projector averages and phase-free groups for
counting. In particular, $H_s\cap H_t$ and $G_s\cap G_t$ are not
interchangeable when eigenvalue signs differ.
\end{remark}

\begin{lemma}[Pauli postselection]\label{lem:postselection}
For a signed Hermitian Pauli $g$, set $Q=(I-g)/2$. If $\ket s$ is a normalized
stabilizer state, then $Q\ket s$ is zero or a scalar multiple of a normalized
stabilizer state. In the latter case its scalar has modulus at most one.
\end{lemma}
\begin{proof}
This is the standard Pauli-measurement update in the stabilizer formalism
\cite{Got97,AG04}. If the outcome is nondeterministic, retain the index-two
subgroup of $H_s$ commuting with $g$ and adjoin $-g$; deterministic outcomes
give zero or $\ket s$. The scalar bound follows from $\norm Q_{\mathrm{op}}=1$.
\end{proof}

\subsection{Coefficient norms and ranks}
\begin{definition}[State and operator ranks and coefficient norms]\label{def:atomic}
For a vector $\ket{v}\in\HH_n$ and an operator $A\in M_N(\C)$,
define the following quantities.
\begin{enumerate}[label=(\roman*)]
\item The \emph{stabilizer rank} is the minimum number of stabilizer
states whose span contains $\ket v$:
\begin{equation}
 \chi(\ket{v})=\min\{k:\ket{v}\in\Span(\ket{s_1},\ldots,\ket{s_k}),\ \ket{s_i}\in\Stab_n\}.
\end{equation}

\item The \emph{stabilizer coefficient norm} is the minimum sum of absolute
coefficients over decompositions into normalized stabilizer states:
\begin{equation}
 \LS(\ket{v})=\min_{\ket{v}=\sum_i c_i\ket{s_i}}\sum_i|c_i|,
 \qquad \ket{s_i}\in\Stab_n.
\end{equation}

\item The \emph{stabilizer extent} is the square of the stabilizer
coefficient norm:
\[
 \xi(\ket{v})=\LS(\ket{v})^2.
\]

\item The \emph{Clifford rank} is the minimum number of Clifford
unitaries whose span contains $A$:
\begin{equation}
 \rC(A)=\min\{k:A\in\Span(C_1,\ldots,C_k),\ C_i\in\Cliff_n\}.
\end{equation}

\item The \emph{Clifford coefficient norm} is the minimum sum of absolute
coefficients over decompositions into Clifford unitaries:
\begin{equation}
 \LC(A)=\min_{A=\sum_i c_iC_i}\sum_i|c_i|,
 \qquad C_i\in\Cliff_n.
\end{equation}
\end{enumerate}
The zero vector and zero operator have rank and coefficient norm zero.
\end{definition}

Global phases of the stabilizer states and Clifford unitaries do not
matter. Both families are finite modulo phase and span their ambient spaces: the computational basis is stabilizer,
and the Paulis form a Clifford operator basis. Combining coefficients of
equal rays reduces the norm minimizations to finite coefficient spaces.
Bounded $\ell_1$ sublevel sets are compact, so the minima exist.
Both unsquared quantities are norms; $\xi$ is homogeneous of degree two.
In particular,
\begin{equation}\label{eq:op-vs-atomic}
 \norm A_{\mathrm{op}}\le\LC(A).
\end{equation}

\begin{remark}[Extent conventions]\label{rem:convention}
We use the squared convention $\xi=\LS^2$ of \cite{BBC19} for states and
retain the unsquared notation $\LC$ for operators. The parameters called
``extent'' in \cite{AD26,DJJMP26,ADGG26} are coefficient $\ell_1$ norms. Every
algorithmic substitution below uses those unsquared parameters explicitly.
\end{remark}

\begin{definition}[Stabilizer fidelity]
For any vector $\ket{v}$, let $\FS(\ket{v})=\max_{\ket s\in\Stab_n}|\ip{s}{v}|^2$.
\end{definition}
\begin{lemma}[Overlap bound from the coefficient norm]\label{lem:overlap}
For all vectors $\ket{\phi},\ket{v}$,
\begin{equation}\label{eq:overlap}
 |\ip{\phi}{v}|\le\sqrt{\FS(\ket{\phi})}\,\LS(\ket{v}).
\end{equation}
In particular, a normalized $\ket{v}$ satisfies $\xi(\ket{v})\FS(\ket{v})\ge1$.
\end{lemma}
\begin{proof}
Apply the triangle inequality to $\ip{\phi}{v}$ in a stabilizer
decomposition of $\ket v$, minimize the coefficient sum, and then take
$\ket\phi=\ket v$ and square to obtain the normalized case \cite{BBC19}.
\end{proof}

\subsection{Approximation metrics}
\begin{definition}[Approximate ranks]\label{def:approx}
\begin{enumerate}[label=(\roman*)]
\item For a normalized state $\ket{\phi}$ and $0\le\delta<1$, define its
\emph{approximate stabilizer rank} $\chi_\delta(\ket{\phi})$ at Euclidean
error $\delta$ by
\[
 \chi_\delta(\ket{\phi})=\min_{\norm{\ket{v}-\ket{\phi}}\le\delta}\chi(\ket{v}).
\]
The approximating vector $\ket v$ need not be normalized.

\item For a unitary $U$ and a Frobenius error tolerance $\delta\ge0$, define its
\emph{approximate Clifford rank} $\rC{}_{,\delta}^{F}(U)$ by
\[
 \rC{}_{,\delta}^{F}(U)
 =\min_{\norm{A-U}_F\le\delta}\rC(A).
\]
The approximating operator $A$ need not be unitary.
\end{enumerate}
\end{definition}

Equivalently, $\chi_\delta(\ket{\phi})$ is the least $k$ for which a
$k$-stabilizer span $V$ obeys $\norm{(I-P_V)\ket{\phi}}\le\delta$.
For normalized pure states,
\[
 \Dtr(\ket{\phi},\ket{\psi})=\tfrac12\norm{\ket\phi\bra\phi-\ket\psi\bra\psi}_1
 =\sqrt{1-|\ip\phi\psi|^2}.
\]
The phase-insensitive Euclidean distance is
$\min_\theta\norm{\ket{\phi}-e^{i\theta}\ket{\psi}}$; it lies between
$\Dtr(\ket{\phi},\ket{\psi})$ and $\sqrt2\Dtr(\ket{\phi},\ket{\psi})$.
For process tomography we use
$\min_\theta\norm{U-e^{i\theta}A}_F$.
Conversion from Choi-vector Euclidean distance to this operator distance
includes a factor of $\sqrt N$. 

\subsection{Function norms and inner products}\label{sec:function-norms}
For functions $f,g:\F_2^n\to\C$, we use the normalized inner product
and norms
\[
 \ip f g=2^{-n}\sum_x\overline{f(x)}g(x),\qquad
 \norm f_2=\sqrt{\ip f f},\qquad
 \norm f_\infty=\max_x|f(x)|.
\]
The support is $\supp(f)=\{x:f(x)\ne0\}$.

\begin{lemma}[Elementary function-norm inequalities]\label{lem:function-norms}
For $f,g:\F_2^n\to\C$, the norms and inner product in
\cref{sec:function-norms} satisfy
\[
 |\ip f g|\le\norm f_2\norm g_2,
 \qquad
 2^{-n}\norm f_\infty^2\le\norm f_2^2\le\norm f_\infty^2.
\]
More precisely,
\begin{equation}\label{eq:function-support-norm}
 \norm f_2^2\le\frac{|\supp(f)|}{2^n}\norm f_\infty^2.
\end{equation}
\end{lemma}
\begin{proof}
Apply Cauchy--Schwarz for the first inequality and bound
$\sum_x|f(x)|^2$ between $\norm f_\infty^2$ and
$|\supp(f)|\norm f_\infty^2$ for the remaining inequalities.
\end{proof}

\section{The stabilizer rank--extent inequality}\label{sec:state-main}

We first prove the rank--coefficient norm comparison for ordinary stabilizer
states.
%This argument has no bipartition and no entanglement parameter.
Its three ingredients are an orthonormal stabilizer basis for each code,
clustering by code overlap, and averaged Pauli projections that eliminate
an entire cluster. The Clifford proof in \cref{sec:operator} follows the
same architecture, with an additional cost for entanglement lost under
postselection.

\subsection{A stabilizer basis for a code}
\begin{lemma}[Dimension bound]\label{lem:state-code-basis}
If $\mathcal C$ is a stabilizer code of dimension $d$, then every
$\ket{v}\in\mathcal C$ satisfies
\begin{equation}\label{eq:state-code-cost}
 \LS(\ket{v})^2\le d\norm{\ket{v}}^2.
\end{equation}
\end{lemma}
\begin{proof}
Extend the code's stabilizer group to a maximal commuting Pauli group. This can be done by adding more independent elements to the list of generators.
Keeping the code's eigenvalues fixed for the original generators, and choosing all possible signs for the
additional generators gives $d$ orthonormal stabilizer states
$\ket{u_1},\ldots,\ket{u_d}$ spanning the code \cite{Got97,Got98}.
Write $\ket v=\sum_{j=1}^d b_j\ket{u_j}$. Since the basis is orthonormal,
$\sum_j|b_j|^2=\norm{\ket v}^2$, and Cauchy--Schwarz gives
\[
 \LS(\ket v)\le\sum_{j=1}^d|b_j|
 \le\sqrt d\left(\sum_{j=1}^d|b_j|^2\right)^{1/2}
 =\sqrt d\norm{\ket v}.
\]
\end{proof}

The next fact permits us to grow a code around several stabilizer states while
controlling its dimension.
\begin{lemma}[Code merging]\label{lem:merge}
Let $\mathcal C_1,\mathcal C_2$ be stabilizer codes on the same qubits,
with dimensions $d_1,d_2$ and projectors $P_1,P_2$. If
$\norm{P_1P_2}_{\mathrm{op}}=\eta>0$, then a stabilizer code containing
both exists with dimension at most $d_1d_2/\eta^2$.
\end{lemma}
\begin{proof}
We construct a containing code by keeping only the stabilizer constraints
shared by the two codes. Let $H_i$ be their canonical signed stabilizer
groups, and set $H=H_1\cap H_2$. Every vector in either code is fixed by
$H$, so its common $+1$ eigenspace $\mathcal C$ contains both codes.

To bound its dimension, expand the projectors using
\eqref{eq:code-projector}. A product of two signed Paulis has nonzero
trace only when the Paulis agree up to sign. Opposite signs cannot occur
here: if $g\in H_1$ and $-g\in H_2$, the two codes lie in orthogonal
eigenspaces of $g$, contradicting $\norm{P_1P_2}_{\mathrm{op}}>0$.
Thus only the matching pairs $g=h\in H$ contribute. For $m$ physical
qubits this gives
\[
 \tr(P_1P_2)
 =\frac{1}{|H_1||H_2|}\sum_{g\in H_1}\sum_{h\in H_2}\tr(gh)
 =\frac{2^m|H|}{|H_1||H_2|}.
\]
Since $d_i=2^m/|H_i|$ and $\dim\mathcal C=2^m/|H|$, it follows that
\[
 \dim\mathcal C=\frac{d_1d_2}{\tr(P_1P_2)}.
\]
Finally, the projector identities $P_i^2=P_i$ and cyclicity of the trace
give
\[
 \tr(P_1P_2)=\tr(P_2P_1P_2)=\norm{P_1P_2}_F^2
 \ge\norm{P_1P_2}_{\mathrm{op}}^2=\eta^2.
\]
Substitution proves the bound; no commutation of $P_1$ and $P_2$ is needed.
\end{proof}

\subsection{Clustering and block elimination}
\begin{lemma}[State clustering dichotomy]\label{lem:state-cluster}
Fix $K\ge2$ and $2\le r\le K$, and put $D=(2K)^{4r}$.
If $\ket{v}$ is a linear combination of at most $K$ normalized stabilizer
states, then either a stabilizer code of dimension at most $D$ contains
at least $r$ of the displayed stabilizer states, or
\begin{equation}\label{eq:state-cluster-cost}
 \LS(\ket{v})^2\le D\norm{\ket{v}}^2.
\end{equation}
\end{lemma}
\begin{proof}
Write the supplied decomposition as $\ket v=\sum_{i=1}^t c_i\ket{s_i}$,
with $t\le K$. The zero vector needs no argument. A \emph{cluster} is a
set of indices from this decomposition, together with a stabilizer code
containing the corresponding stabilizer states. Begin with singleton clusters, whose
codes are the one-dimensional rays of the stabilizer states. Whenever two cluster
projectors satisfy $\norm{P_iP_j}_{\mathrm{op}}\ge(2K)^{-1}$, merge
their index sets and use \cref{lem:merge} with
$\eta=\norm{P_iP_j}_{\mathrm{op}}$ to choose a containing code;
the dimension increases by a factor of at most $\eta^{-2}\le(2K)^2$
relative to the product of the two code dimensions.

Throughout this procedure, a cluster of $q$ indices has code dimension at
most $(2K)^{2(q-1)}$. Indeed, the bound is one for a singleton, and merging
clusters of sizes $q_1,q_2$ gives dimension at most
\[
 (2K)^2(2K)^{2(q_1-1)}(2K)^{2(q_2-1)}
 =(2K)^{2(q_1+q_2-1)}.
\]
Stop if a cluster first reaches size at least $r$. Both clusters merged
at that step had size at most $r-1$, so the new cluster has at most
$2r-2$ indices. Its code dimension is at most $(2K)^{4r-6}\le D$,
giving the first alternative.

If this never happens, the procedure ends with clusters
$I_1,\ldots,I_h$, each of size less than $r$. Define their contributions
using the original coefficients:
\[
 \ket{v_j}=\sum_{i\in I_j}c_i\ket{s_i},\qquad
 \ket v=\sum_{j=1}^h\ket{v_j}.
\]
Their code dimensions satisfy $d_j\le(2K)^{2(r-2)}$, and $h\le K$.
Since no further merge is possible, vectors in distinct cluster codes
have small inner product:
\[
 |\ip{v_i}{v_j}|=|\ip{v_i}{P_iP_jv_j}|
 \le(2K)^{-1}\norm{\ket{v_i}}\norm{\ket{v_j}}
 \qquad(i\ne j).
\]
In particular, the cluster contributions cannot cancel too much. Writing
$x_j=\norm{\ket{v_j}}$ and using $2x_ix_j\le x_i^2+x_j^2$, we obtain
\[
 \norm{\ket v}^2
 \ge\sum_jx_j^2-\frac1K\sum_{i<j}x_ix_j
 \ge\left(1-\frac{h-1}{2K}\right)\sum_jx_j^2
 \ge\frac12\sum_jx_j^2.
\]
We can therefore decompose each contribution in its own code basis and
combine the resulting coefficient bounds. By
\cref{lem:state-code-basis}, the triangle inequality, and Cauchy--Schwarz,
\[
 \LS(\ket{v})^2
 \le\left(\sum_j\sqrt{d_j}\norm{\ket{v_j}}\right)^2
 \le\left(\sum_jd_j\right)\left(\sum_j\norm{\ket{v_j}}^2\right)
 \le2K(2K)^{2(r-2)}\norm{\ket{v}}^2\le D\norm{\ket{v}}^2.
\]
\end{proof}

\begin{lemma}[State block elimination]\label{lem:state-block}
Suppose a code of dimension at most $D$ contains at least $r$ stabilizer states in
a $t$-term stabilizer decomposition of $\ket{v}$. Suppose every vector with
at most $t-r$ stabilizer summands satisfies $\LS(\ket{w})^2\le b\norm{\ket{w}}^2$.
Then
\[
 \LS(\ket{v})^2\le(D+2b)\norm{\ket{v}}^2.
\]
\end{lemma}
\begin{proof}
Split $\ket v$ into its component in the selected code and its orthogonal
remainder. The code-basis bound controls the first component. To control
the remainder, we express it as an average of projections that remove
the selected $r$ stabilizer states.

Let $H$ be the code's signed stabilizer group, let $P$ be its projector,
and put $Q_g=(I-g)/2$ for $g\in H$. Every selected stabilizer state is fixed by $g$,
so $Q_g$ sends it to zero. By \cref{lem:postselection}, the other terms
remain scalar multiples of stabilizer states or become zero. Thus
$Q_g\ket v$ has at most $t-r$ stabilizer summands, and the assumption of
the lemma gives $\LS(Q_g\ket v)\le\sqrt b\norm{Q_g\ket v}$.

Here $\E_g$ denotes the uniform average over $H$. The projector formula
$\E_g g=P$ gives $\E_gQ_g=(I-P)/2$. Also $Q_g^2=Q_g$, so
$\E_g\norm{Q_g\ket v}^2=\ip{v}{(\E_gQ_g)v}$. These observations yield
\begin{equation}\label{eq:state-average-identities}
 \ket{v}=P\ket{v}+2\E_gQ_g\ket{v},
 \qquad \E_g\norm{Q_g\ket{v}}^2=\tfrac12\norm{(I-P)\ket{v}}^2.
\end{equation}
Apply \cref{lem:state-code-basis} to $P\ket v$ and the preceding
coefficient-norm bound to each $Q_g\ket v$. The triangle inequality and
$\E_g\norm{Q_g\ket v}\le(\E_g\norm{Q_g\ket v}^2)^{1/2}$ give
\begin{align*}
 \LS(\ket{v})
 &\le\sqrt D\norm{P\ket{v}}+2\sqrt b\,\E_g\norm{Q_g\ket{v}}\\
 &\le\sqrt D\norm{P\ket{v}}+\sqrt{2b}\norm{(I-P)\ket{v}}\\
 &\le\sqrt{D+2b}\norm{\ket{v}}.
\end{align*}
The last step is Cauchy--Schwarz for the two terms, together with
$\norm{P\ket v}^2+\norm{(I-P)\ket v}^2=\norm{\ket v}^2$.
\end{proof}

\begin{theorem}[Finite stabilizer rank--extent bound]
\label{thm:state-finite-main}
If $\chi(\ket{v})\le k$, then for $k\ge2$ and every $2\le r\le k$,
\begin{equation}\label{eq:state-finite-bound}
 \xi(\ket{v})\le(2k)^{4r}\left(2^{\floor{k/r}+1}-1\right)\norm{\ket{v}}^2.
\end{equation}
For all $k\ge1$, also $\xi(\ket{v})\le(2^k-1)\norm{\ket{v}}^2$.
\end{theorem}
\begin{proof}
Fix $K=k$, $r$, and $D=(2K)^{4r}$. Clustering either gives
$\LS(\ket v)^2\le D\norm{\ket v}^2$ directly or supplies a block of
$r$ stabilizer states that can be eliminated. By
\cref{lem:state-block}, a bound $b$ for the remaining terms then gives
the bound $D+2b$. This suggests the following sequence, indexed by the
number of terms:
\[
 b_t=D\sum_{j=0}^{\floor{t/r}}2^j
    =D\left(2^{\floor{t/r}+1}-1\right),\qquad 0\le t\le K.
\]
It satisfies $b_t=D$ for $t<r$ and $b_t=D+2b_{t-r}$ for $t\ge r$.

We prove by induction on $t$ that every vector with at most $t$
stabilizer summands obeys $\LS(\ket v)^2\le b_t\norm{\ket v}^2$.
The case $t=0$ is immediate. If $0<t<r$, a cluster of $r$ supplied
stabilizer states is impossible, so \cref{lem:state-cluster} gives the bound $D=b_t$.
For $t\ge r$, the same lemma either gives the bound $D\le b_t$ or
provides a code containing at least $r$ stabilizer states. In the latter case,
induction bounds every vector with at most $t-r$ summands by $b_{t-r}$,
and \cref{lem:state-block} gives $D+2b_{t-r}=b_t$.
Taking $t=k$ and recalling $\xi=\LS^2$ proves
\eqref{eq:state-finite-bound}.

For the second estimate, eliminate one stabilizer state at a time. Its ray is a
one-dimensional code, so \cref{lem:state-block} gives the recurrence
$c_0=0$ and $c_t=1+2c_{t-1}$. Hence $c_t=2^t-1$, as claimed.
\end{proof}

\begin{corollary}[Uniform stabilizer rank--extent bound]
\label{cor:state-envelope}
Define $\Snd(0)=0$ and, for integers $k\ge1$, set
\[
 \Snd(k)=2^{12\sqrt{k\log_2(2k)}}.
\]
Then $\chi(\ket{v})\le k$ implies $\xi(\ket{v})\le\Snd(k)\norm{\ket{v}}^2$.
\end{corollary}
This implies \cref{thm:intro-state}.
The logarithm of the bound in \eqref{eq:state-finite-bound} is at most
$4rL+k/r+1$, where $L=\log_2(2k)$; choosing $r$ of order
$\sqrt{k/L}$ balances the block cost against the number of elimination
steps. The detailed estimates are given in \cref{sec:app-uniform-bounds}.

\section{The Clifford rank--coefficient norm inequality}\label{sec:operator}

We combine the code-clustering and block-elimination argument of
\cref{sec:state-main} with the entanglement-control ingredients used by
\cite{DJJMP26} to obtain the Clifford rank--coefficient norm bound.
The Clifford rank--norm bound in \cite[Corollary~4.5]{DJJMP26} follows
directly from the coefficient bound of \cite{KS26} applied to the original
Clifford Choi decomposition. Our state extent bound instead allows a new
stabilizer decomposition, whose terms need not be maximally entangled,
so it does not give the improved Clifford bound as a black-box consequence.
We can, however, use the entanglement-deficit bound of
\cite[Lemma~3.4]{DJJMP26} directly within our proof: each Pauli projection
increases the deficit by at most one. The code-merging step is unchanged,
and we track this deficit through block elimination to control the cost
of decomposing the resulting states into maximally entangled stabilizer
states.

\subsection{The normalized Choi map}
For $A\in M_N(\C)$, define
\begin{equation}\label{eq:choi}
 \J A=(A\otimes I)\ket{\Phi_N},\qquad
 \ket{\Phi_N}=N^{-1/2}\sum_x\ket{x,x}.
\end{equation}
With this normalization,
\begin{equation}\label{eq:choi-isometry}
 \norm{\J A}^2=\frac{\norm A_F^2}{N}.
\end{equation}
Thus the map is normalized, but $\J A$ need not be a unit vector.

Let $\mathcal M_n\subseteq\Stab_{2n}$ be the stabilizer states maximally
entangled across the first $n$ versus last $n$ qubits, and let $\LM$ be
the minimum coefficient $\ell_1$-norm over decompositions into states
from $\mathcal M_n$, defined on $\HH_{2n}$. The bipartite stabilizer normal form
\cite{FCYBC04} identifies every such state with
$(C_1\otimes C_2)\ket{\Phi_N}$ for local Cliffords. Since
$(C_1\otimes C_2)\ket{\Phi_N}=\J{C_1C_2^T}$ and the transpose of a
Clifford is Clifford, these are precisely Clifford Choi states. Hence
\begin{equation}\label{eq:choi-atomic}
 \LM(\J A)=\LC(A).
\end{equation}
The norm is defined on all vectors: maximally entangled Pauli Choi states
already form an orthonormal basis.

\subsection{Entanglement deficit and stabilizer decompositions}
\begin{definition}[Deficit]\label{def:deficit}
For a normalized stabilizer state $\ket s$ on the balanced $n:n$ cut,
let $\SR(\ket s)$ denote its Schmidt rank across this cut. Define its \emph{entanglement deficit} $a(\ket s)$, the number of ebits
missing from maximal entanglement, by
\[
 a(\ket s)=n-\log_2\SR(\ket s).
\]
Let $\PG_A\le\PG_{2n}$ be the subgroup of phase-free Paulis supported
only on the first subsystem. Then
\[
 a(\ket s)=\log_2|G_s\cap\PG_A|.
\]
In particular, the deficit is an integer between zero and $n$.
\end{definition}

For completeness, let
$H_{s,A}=\{P_A:P_A\otimes I\in H_s\}$. Tracing the stabilizer expansion
over the second subsystem gives
\[
 \rho_A=2^{-n}\sum_{g\in H_{s,A}}g
       =\frac{|H_{s,A}|}{2^n}P_{H_{s,A}},
 \qquad \rank\rho_A=\frac{2^n}{|H_{s,A}|}.
\]
Here $P_{H_{s,A}}$ is the projector onto the common $+1$ eigenspace of
$H_{s,A}$, and $|H_{s,A}|=|G_s\cap\PG_A|$. The nonzero eigenvalues
of $\rho_A$ are uniform. Since $\SR(\ket s)=\rank\rho_A$, this proves the
formula for the deficit and its integrality. The bipartite normal form is
\begin{equation}\label{eq:normal-form}
 \ket s\ \sim_{\text{local Clifford}}\quad
 \ket{\Phi_2}^{\otimes(n-a)}\otimes
 \ket0_A^{\otimes a}\ket0_B^{\otimes a},\qquad a=a(\ket s)
\end{equation}
\cite{FCYBC04}.

\begin{lemma}[Decomposition into maximally entangled stabilizer states]\label{lem:repair}
Every normalized stabilizer state $\ket s$ satisfies
\begin{equation}\label{eq:repair}
 \LM(\ket s)\le2^{a(\ket s)/2}.
\end{equation}
\end{lemma}
\begin{proof}
Put $a=a(\ket s)$, the number of missing ebits. In the normal form
\eqref{eq:normal-form}, replace each of the $a$ unentangled pairs by
\[
 \ket{00}=\frac{\ket{\Phi^+}+\ket{\Phi^-}}{\sqrt2},\qquad
 \ket{\Phi^\pm}=\frac{\ket{00}\pm\ket{11}}{\sqrt2}.
\]
The resulting $2^a$ terms are maximally entangled stabilizer states, each
with coefficient $2^{-a/2}$. Their total coefficient norm is therefore
$2^a2^{-a/2}=2^{a/2}$. Undoing the local Clifford transformation preserves
both maximal entanglement and the stabilizer property, so this is also
a valid decomposition for $\ket s$.
\end{proof}

\begin{lemma}[One measurement loses at most one ebit]\label{lem:deficit-step}
If $Q=(I-g)/2$ is a Pauli projector and $Q\ket s\ne0$, then
\[
 a(Q\ket s/\norm{Q\ket s})\le a(\ket s)+1.
\]
\end{lemma}
\begin{proof}
In a nondeterministic outcome, $\ket s$ and the normalized projection
have squared overlap $1/2$ and are neighboring stabilizer states. Their
entanglement entropies differ by at most one by
\cite[Lemma~3.4]{DJJMP26}; for a bipartite stabilizer state this entropy
is $n-a$. A deterministic nonzero outcome leaves $\ket s$ unchanged.
\end{proof}

\begin{lemma}[A low-deficit basis of a code]\label{lem:code-basis}
Suppose a stabilizer code $\mathcal C$ of dimension $d$ contains a
normalized stabilizer $\ket s$ with $a(\ket s)\le a$. Then every $\ket{v}\in\mathcal C$
satisfies
\begin{equation}\label{eq:code-cost}
 \LM(\ket{v})^2\le2^a d\norm{\ket{v}}^2.
\end{equation}
\end{lemma}
\begin{proof}
We need a code basis whose vectors have the same deficit as the given
$\ket s$, not just any stabilizer basis. Let $d=2^\ell$. Since
$\ket s\in\mathcal C$, the code's signed stabilizer group is a subgroup
of $H_s$. Choose generators of $H_s$ extending those of the code.
By \cref{lem:clifford-conjugation}, a Clifford $F$ sends these generators
to the positive single-qubit $Z$ operators, with the code generators
first. Up to global phase,
\[
 F\ket s=\ket{0^{2n}},\qquad
 F\mathcal C=\ket{0^{2n-\ell}}\otimes\HH_\ell.
\]
Thus $\ket{u_y}=F^\dagger\ket{0^{2n-\ell},y}$, for
$y\in\F_2^\ell$, form an orthonormal stabilizer basis of the code.
Each is a Pauli image of $\ket s$, because
\[
 \ket{u_y}=\bigl(F^\dagger(I\otimes X^y)F\bigr)\ket s.
\]
Although $F$ need not be local across the balanced cut, the Pauli in
parentheses is a tensor product of operators on the two halves. It
therefore preserves entanglement, so every $\ket{u_y}$ has deficit
$a(\ket s)\le a$.

Write $\ket v=\sum_y b_y\ket{u_y}$. Using \cref{lem:repair} for each
basis vector and then Cauchy--Schwarz gives
\[
 \LM(\ket v)\le\sum_y|b_y|\LM(\ket{u_y})
 \le2^{a/2}\sum_y|b_y|
 \le2^{a/2}\sqrt d\norm{\ket v}.
\]
\end{proof}

\subsection{Clustering with an entanglement weight}
The code-merging argument of \cref{lem:state-cluster} does not depend on
the bipartition. Its terminal-cluster estimate does: here we use
\cref{lem:code-basis} in place of \cref{lem:state-code-basis}.
\begin{lemma}[Weighted clustering dichotomy]\label{lem:cluster}
Fix $K\ge2$ and $2\le r\le K$, and put $D=(2K)^{4r}$.
Suppose $\ket{v}=\sum_{i=1}^t c_i\ket{s_i}$, where $t\le K$ and each $\ket{s_i}$ is a
normalized balanced-cut stabilizer of deficit at most $a$. Then either
a stabilizer code of dimension at most $D$ contains at least $r$ of the
displayed stabilizer states, or
\[
 \LM(\ket{v})^2\le2^aD\norm{\ket{v}}^2.
\]
\end{lemma}
\begin{proof}
Run the same merging procedure as in \cref{lem:state-cluster}. The code
dimension invariant and the bound $D$ for a cluster of at least $r$
stabilizer states are unchanged. If no such cluster appears, write $\ket{v}=\sum_j\ket{v_j}$
over the terminal clusters, whose code dimensions are $d_j$.
Their pairwise projector overlaps give
$\sum_j\norm{\ket{v_j}}^2\le2\norm{\ket{v}}^2$, while
$\sum_jd_j\le K(2K)^{2(r-2)}$. Each cluster code contains one of its
original stabilizer states, of deficit at most $a$. Thus \cref{lem:code-basis}
and Cauchy--Schwarz give
\[
 \LM(\ket{v})^2
 \le2^a\left(\sum_j\sqrt{d_j}\norm{\ket{v_j}}\right)^2
 \le2^a\left(\sum_jd_j\right)
          \left(\sum_j\norm{\ket{v_j}}^2\right)
 \le2^aD\norm{\ket{v}}^2.
\]
\end{proof}

\subsection{The weighted block recurrence}
\begin{lemma}[Eliminating a block]\label{lem:block}
Suppose a code $\mathcal C$ of dimension at most $D$ contains at least $r$
stabilizer states in a $t$-term decomposition of $\ket{v}$ into stabilizers of deficit at most
$a$. Suppose every vector $\ket{w}$ expressible using at most $t-r$ stabilizers
of deficit at most $a+1$ satisfies
\[
 \LM(\ket{w})^2\le2^{a+1}b\norm{\ket{w}}^2.
\]
Then $\LM(\ket{v})^2\le2^a(D+4b)\norm{\ket{v}}^2$.
\end{lemma}
\begin{proof}
We use the same decomposition as in \cref{lem:state-block}, but account
for the possible loss of one ebit under projection. Let $P=P_{\mathcal C}$
and $Q_g=(I-g)/2$ for $g\in H_{\mathcal C}$. Each $Q_g$ removes the
selected stabilizer states. By \cref{lem:postselection,lem:deficit-step},
$Q_g\ket v$ is a combination of at most $t-r$ stabilizers of deficit at
most $a+1$. The assumption therefore bounds its coefficient norm by
$\sqrt{2^{a+1}b}\norm{Q_g\ket v}$.

Let $\E_g$ be the uniform average over the signed group $H_{\mathcal C}$.
The identities derived in \eqref{eq:state-average-identities} apply
unchanged:
\begin{equation}\label{eq:average-identities}
 \ket{v}=P\ket{v}+2\E_gQ_g\ket{v},\qquad
 \E_g\norm{Q_g\ket{v}}^2=\tfrac12\norm{(I-P)\ket{v}}^2.
\end{equation}
The code contains an original stabilizer state of deficit at most $a$, so
\cref{lem:code-basis} bounds the coefficient norm of $P\ket v$ by
$2^{a/2}\sqrt D\norm{P\ket v}$. Combining the two bounds and applying
Cauchy--Schwarz to the average gives
\begin{align*}
 \LM(\ket{v})
 &\le 2^{a/2}\sqrt D\norm{P\ket{v}}
       +2\sqrt{2^{a+1}b}\,\E_g\norm{Q_g\ket{v}}\\
 &\le 2^{a/2}\left(\sqrt D\norm{P\ket{v}}
                  +2\sqrt b\norm{(I-P)\ket{v}}\right)\\
 &\le2^{a/2}\sqrt{D+4b}\norm{\ket{v}}.
\end{align*}
The last inequality uses the orthogonality of $P\ket v$ and
$(I-P)\ket v$, just as in \cref{lem:state-block}. In that lemma the
recurrence was $D+2b$; here it is $D+4b$. The extra factor of two in
the squared recursive cost comes from the possible increase in deficit
from $a$ to $a+1$.
\end{proof}

\begin{theorem}[Finite rank--norm bound]\label{thm:finite-main}
Let $A$ be an operator on $n$ qubits and $\rC(A)\le k$.
For $k\ge2$ and every integer $2\le r\le k$,
\begin{equation}\label{eq:finite-bound}
 \LC(A)^2\le
 (2k)^{4r}\frac{4^{\floor{k/r}+1}-1}{3}\,2^{-n}\norm A_F^2.
\end{equation}
For every $k\ge1$, the additional estimate
\begin{equation}\label{eq:single-bound}
 \LC(A)^2\le\frac{4^k-1}{3}2^{-n}\norm A_F^2
\end{equation}
holds. The zero operator satisfies both estimates trivially.
\end{theorem}
\begin{proof}
First we prove a bound for stabilizer combinations at every deficit;
then we apply its deficit-zero case to the Choi vector of $A$. Fix
$K=k$, $r$, and $D=(2K)^{4r}$. The block-elimination step now has
recurrence $b_t=D+4b_{t-r}$, so define
\[
 b_t=D\sum_{j=0}^{\floor{t/r}}4^j,\qquad 0\le t\le K.
\]
We induct on $t$, proving simultaneously for all nonnegative integers
$a$ that every vector
$\ket{v}$ with at most $t$ stabilizer summands of deficit at most $a$ obeys
$\LM(\ket{v})^2\le2^ab_t\norm{\ket{v}}^2$. We include all $a$ because
eliminating a block reduces the number of terms but may increase the
deficit. Values $a>n$ merely give weaker bounds on the same class of stabilizer states.

For $t<r$, the first alternative of \cref{lem:cluster} is impossible,
and the second alternative gives the claim. This includes $t=0$,
when $\ket{v}=0$. For $t\ge r$, either the explicit clustering bound gives
$2^aD\norm{\ket{v}}^2$, or \cref{lem:block} and induction give
$2^a(D+4b_{t-r})\norm{\ket{v}}^2=2^ab_t\norm{\ket{v}}^2$.

A rank-$k$ Clifford decomposition of $A$ gives a $k$-term decomposition of
$\J A$ into deficit-zero stabilizers. Apply the conclusion with $a=0$,
then use \eqref{eq:choi-isometry} and \eqref{eq:choi-atomic}. This proves
\eqref{eq:finite-bound}.

For \eqref{eq:single-bound}, eliminate one stabilizer state at a time, using its
one-dimensional ray as the code. The same induction gives
$c_0=0$ and $c_t=1+4c_{t-1}=(4^t-1)/3$. Applying the deficit-zero
case to $\J A$ proves the second bound.
\end{proof}

\begin{definition}[A uniform bound function]\label{def:envelope}
The function $\Bnd(k)$ is the monotone bound function used in the
Clifford rank--norm inequality below.
Set $\Bnd(0)=0$ and, for every integer $k\ge1$, define
\begin{equation}\label{eq:envelope}
 \Bnd(k)=2^{12\sqrt{k\log_2(2k)}}.
\end{equation}
The constant twelve is convenient rather than optimized.
\end{definition}

\begin{corollary}[Uniform operator inequality]\label{cor:main-envelope}
If $\rC(A)\le k$, then
\begin{equation}\label{eq:operator-envelope}
 \LC(A)^2\le\Bnd(k)2^{-n}\norm A_F^2.
\end{equation}
\end{corollary}
This implies \cref{thm:intro-main}.
As in the state case, choosing $r$ of order $\sqrt{k/L}$, with
$L=\log_2(2k)$, balances the terms $4rL$ and $2k/r$ in the logarithm
of \eqref{eq:finite-bound}. The detailed estimates are given in
\cref{sec:app-uniform-bounds}.

\section{A nearly quadratic lower bound for AND}\label{sec:uncertainty}

\begin{definition}[Stable rank {\cite{MZ10}}]\label{def:stable-rank}
For $A\ne0$, define its \emph{stable rank} by
$\sr(A)=\norm A_F^2/\norm A_{\mathrm{op}}^2$.
It satisfies $1\le\sr(A)\le\rank A$.
\end{definition}

\begin{theorem}[Clifford stable-rank uncertainty]\label{thm:stable-rank}
If $A\ne0$ is an $n$-qubit operator of Clifford rank at most $k$, then
\begin{equation}\label{eq:stable-rank}
 \sr(A)\ge\frac{2^n}{\Bnd(k)}.
\end{equation}
In particular, $\rank A\ge2^n/\Bnd(k)$.
\end{theorem}
\begin{proof}
Since $A\ne0$, its operator norm is positive. By
\eqref{eq:op-vs-atomic} and \cref{cor:main-envelope},
\[
 \norm A_{\mathrm{op}}^2\le\LC(A)^2
 \le\Bnd(k)2^{-n}\norm A_F^2.
\]
Rearranging gives
$\sr(A)=\norm A_F^2/\norm A_{\mathrm{op}}^2\ge2^n/\Bnd(k)$.
The ordinary-rank bound follows from $\rank A\ge\sr(A)$.
\end{proof}

\begin{definition}[Quadratic-phase rank]\label{def:quadratic}
For $f:\F_2^n\to\C$, define its quadratic-phase rank as follows. All polynomials
below are Boolean multilinear polynomials over $\F_2$, and coefficients
$c_j$ may be complex.
The \emph{quadratic-phase rank} is the minimum number of quadratic
phases needed to express $f$ as a linear combination:
\begin{equation}
 R_2(f)=\min\left\{k:f=\sum_{j=1}^k c_j(-1)^{q_j},\quad
                                      \deg q_j\le2\right\}.
\end{equation}
Set $\QAND(n)=R_2(\AND_n)$, where $\AND_n(x)=\prod_i x_i$.
Flipping every input bit exchanges AND and NOR and preserves quadratic-phase
rank. Hence $\QAND(n)=R_2(\NOR_n)$, where
\[
 \NOR_n(x)=\begin{cases}1,&x=0,\\0,&x\ne0.\end{cases}
\]
\end{definition}

A quadratic phase is the diagonal of a Clifford generated by $Z$ and $CZ$.
Thus, for $D_f=\diag(f(x))$, using the function norms from
\cref{sec:function-norms},
\begin{equation}\label{eq:diagonal-lift}
 \rC(D_f)\le R_2(f),\qquad
 \norm{D_f}_F=\sqrt N\norm f_2,\qquad
 \norm{D_f}_{\mathrm{op}}=\norm f_\infty.
\end{equation}

\begin{corollary}[Support uncertainty]\label{cor:support}
For every nonzero $f$ with $R_2(f)\le k$, using the function norms from
\cref{sec:function-norms},
\begin{equation}\label{eq:support}
 \norm f_2^2\ge\frac{\norm f_\infty^2}{\Bnd(k)},\qquad
 |\supp f|\ge\frac{2^n}{\Bnd(k)}.
\end{equation}
\end{corollary}
\begin{proof}
By \eqref{eq:diagonal-lift}, $D_f$ has Clifford rank at most $k$ and
stable rank $2^n\norm f_2^2/\norm f_\infty^2$.
Applying \cref{thm:stable-rank} gives the first inequality.
Combine it with \eqref{eq:function-support-norm} from
\cref{lem:function-norms} to get
\[
 \frac{\norm f_\infty^2}{\Bnd(k)}
 \le\norm f_2^2
 \le\frac{|\supp(f)|}{2^n}\norm f_\infty^2.
\]
Since $f\ne0$, cancel $\norm f_\infty^2>0$ to obtain the support bound.
\end{proof}

\begin{lemma}[Logarithmic inversion]\label{lem:inversion}
For $k,L\ge1$,
\[
 L^2=O\!\bigl(k\log(1+k)\bigr)
 \quad\Longrightarrow\quad
 k=\Omega\!\left(\frac{L^2}{\log(1+L)}\right).
\]
\end{lemma}
The proof is recorded in \cref{sec:app-inversion}. We use the equivalent
forms $\log(2k)$ and $\log(1+k)$ interchangeably in asymptotic estimates.

\begin{corollary}[Exact quadratic uncertainty for AND]\label{cor:and}
As $n\to\infty$,
\begin{equation}\label{eq:and-lower}
 \QAND(n)=\Omega\!\left(\frac{n^2}{\log n}\right).
\end{equation}
Every nonzero rank-one operator on $n$ qubits also has Clifford rank
$\Omega(n^2/\log n)$.
\end{corollary}
\begin{proof}
Let $k=R_2(\AND_n)=\QAND(n)$, the minimum number of quadratic phases
needed to represent $\AND_n$ (\cref{def:quadratic}). Since $\AND_n$ is
nonzero only at the all-ones input, $|\supp(\AND_n)|=1$.
Applying \cref{cor:support} gives $1\ge2^n/\Bnd(k)$, hence
$\Bnd(k)\ge2^n$. By \eqref{eq:envelope}, this means
\[
 2^{12\sqrt{k\log_2(2k)}}\ge2^n,
 \qquad k\log_2(2k)\ge n^2/144.
\]
Apply \cref{lem:inversion} with $L=n$ to obtain
$k=\Omega(n^2/\log n)$.

For the operator claim, let $A$ be a nonzero rank-one operator and set
$k=\rC(A)$, its Clifford rank. There is only one nonzero singular value,
so $\norm A_F=\norm A_{\mathrm{op}}$ and $\sr(A)=1$
by \cref{def:stable-rank}. Thus \cref{thm:stable-rank} gives
$1\ge2^n/\Bnd(k)$. The same calculation and inversion yield
$\rC(A)=\Omega(n^2/\log n)$.
\end{proof}

The AND bound concerns exact representation: the zero function already
approximates $\AND_n$ with squared $L_2$ error $2^{-n}$, so constant
absolute error cannot give a comparable lower bound.

\section{Fidelity bounds and stabilizer-rank lower bounds}\label{sec:state}

Throughout this section, $\Snd$ is the state bound function from
\cref{cor:state-envelope}.
\begin{theorem}[Uniform span and robust fidelity bounds]\label{thm:span}
Let $V$ be spanned by at most $k$ stabilizer states, and let $P_V$ denote
the orthogonal projector onto $V$. Then, for every vector
$\ket{\phi}$,
\begin{equation}\label{eq:span-bound}
 \norm{P_V\ket{\phi}}^2\le\Snd(k)\FS(\ket{\phi}).
\end{equation}
If $\ket{\phi}$ is normalized and $\chi_\delta(\ket{\phi})\le k$ with $0\le\delta<1$, then
\begin{equation}\label{eq:robust-fidelity}
 \FS(\ket{\phi})\ge\frac{1-\delta^2}{\Snd(k)}.
\end{equation}
\end{theorem}
\begin{proof}
We first bound the projection of $\ket\phi$ onto $V$. If
$P_V\ket\phi=0$, the bound is immediate. Otherwise, normalize this
projection to obtain
\[
 \ket v=\frac{P_V\ket\phi}{\norm{P_V\ket\phi}}.
\]
Since $V$ is spanned by at most $k$ stabilizer states, $\ket v$ has
stabilizer rank at most $k$. Its stabilizer extent therefore satisfies
$\xi(\ket v)\le\Snd(k)$ by \cref{cor:state-envelope}.
The overlap bound in \cref{lem:overlap}, together with
$\xi(\ket v)=\LS(\ket v)^2$, now gives
\[
 \norm{P_V\ket\phi}^2=|\ip\phi v|^2
 \le\FS(\ket\phi)\xi(\ket v)
 \le\Snd(k)\FS(\ket\phi).
\]
Here $\FS(\ket\phi)$ is the largest squared overlap of $\ket\phi$
with a normalized stabilizer state. This proves \eqref{eq:span-bound}.

For the second claim, suppose $\ket\phi$ is normalized and
$\chi_\delta(\ket\phi)\le k$. By \cref{def:approx}, there is a vector
$\ket w$ of stabilizer rank at most $k$ with
$\norm{\ket\phi-\ket w}\le\delta$. Let $V$ be the span of the at most
$k$ stabilizer states in a decomposition of $\ket w$.
The orthogonal projection $P_V\ket\phi$ is the closest vector in $V$
to $\ket\phi$, so its error is also at most $\delta$. Pythagoras gives
\[
 \norm{P_V\ket\phi}^2
 =1-\norm{\ket\phi-P_V\ket\phi}^2\ge1-\delta^2.
\]
Combining this with \eqref{eq:span-bound} and dividing by $\Snd(k)$
proves \eqref{eq:robust-fidelity}.
\end{proof}

\begin{corollary}[Fidelity forces nearly quadratic rank]\label{cor:rank-from-fid}
There is a universal $c>0$ such that, for a normalized state $\ket\phi$
and $0\le\delta<1$, if
\[
 L=\log_2\frac{1-\delta^2}{\FS(\ket{\phi})}\ge1,
\]
then
\begin{equation}\label{eq:rank-from-fid}
 \chi_\delta(\ket{\phi})\ge c\frac{L^2}{\log_2(1+L)}.
\end{equation}
More generally, if normalized $\ket{v},\ket{\phi}$ satisfy
$|\ip\phi v|^2\ge\alpha$ and
$L=\log_2(\alpha/\FS(\ket{\phi}))\ge1$, the same lower bound holds for $\chi(\ket{v})$.
\end{corollary}
\begin{proof}
In the first case, set $k=\chi_\delta(\ket\phi)$, the approximate
stabilizer rank from \cref{def:approx}. The robust fidelity bound in
\cref{thm:span} gives $\Snd(k)\ge(1-\delta^2)/\FS(\ket\phi)=2^L$.
In the second case, set $k=\chi(\ket v)$ and choose a span $V$ of $k$
stabilizers containing $\ket v$. Since $\ket v$ is normalized,
$\alpha\le|\ip\phi v|^2\le\norm{P_V\ket\phi}^2$.
The span bound in \cref{thm:span} therefore gives
$\Snd(k)\ge\alpha/\FS(\ket\phi)=2^L$. In both cases $k\ge1$, and
\[
 \log_2\Snd(k)=12\sqrt{k\log_2(2k)}.
\]
Thus $k\log_2(2k)\ge L^2/144$; apply \cref{lem:inversion}.
\end{proof}

\begin{proposition}[Fidelity multiplicativity {\cite{BBC19}}]\label{prop:alignment}
If each normalized pure state $\ket{\phi_j}$ acts on at most three qubits, then
\[
 \FS\!\left(\bigotimes_j\ket{\phi_j}\right)=\prod_j\FS(\ket{\phi_j}).
\]
\end{proposition}

\begin{corollary}[Magic-state tensor powers]\label{cor:magic-rank}
Let $\ket{\phi}\in\C^2$ be a fixed normalized non-stabilizer state. For every
fixed $0\le\delta<1$,
\[
 \chi_\delta(\ket{\phi}^{\otimes n})
 =\Omega_{\phi,\delta}(n^2/\log n).
\]
In particular, for $\ket T=(\ket0+e^{i\pi/4}\ket1)/\sqrt2$,
\begin{equation}\label{eq:T-rank}
 \chi_\delta(\ket{T}^{\otimes n})=\Omega_\delta(n^2/\log n).
\end{equation}
\end{corollary}
\begin{proof}
Since $\ket\phi$ is not a stabilizer and the family of one-qubit stabilizer
states is finite, $\gamma=\FS(\ket\phi)$ satisfies $0<\gamma<1$.
By \cref{prop:alignment}, $\FS(\ket{\phi}^{\otimes n})=\gamma^n$.
For fixed $\delta<1$, the parameter in \cref{cor:rank-from-fid} is
\[
 L=\log_2\frac{1-\delta^2}{\gamma^n}
   =n\log_2(1/\gamma)+\log_2(1-\delta^2).
\]
The first term grows linearly in $n$ and the second is constant.
Thus $L=\Theta_{\phi,\delta}(n)$, and the corollary gives the stated
$\Omega_{\phi,\delta}(n^2/\log n)$ lower bound for sufficiently large $n$.
For $\ket{T}$, direct comparison with the six Pauli eigenstates gives
$\gamma=\cos^2(\pi/8)$.
\end{proof}
For the $T$ state, a nearly quadratic constant-error bound was already
proved in \cite[Theorem~1.1]{MT24}. The point here is the explicit
$\log n$ loss obtained uniformly for every fixed one-qubit
non-stabilizer state from the rank--extent inequality.

\begin{corollary}[Exponentially small target overlap]\label{cor:tiny-overlap}
Suppose normalized targets $\ket{\phi_n}$ have $\FS(\ket{\phi_n})\le q^n$ and
normalized $\ket{v_n}$ satisfy $|\ip{\phi_n}{v_n}|^2\ge p^n$, where
$0<q<p<1$ are fixed. Then $\chi(\ket{v_n})=\Omega_{p,q}(n^2/\log n)$.
\end{corollary}
\begin{proof}
Apply the overlap version of \cref{cor:rank-from-fid} with $\alpha=p^n$.
Its parameter satisfies
\[
 L=\log_2\frac{p^n}{\FS(\ket{\phi_n})}
 \ge\log_2\frac{p^n}{q^n}=n\log_2(p/q).
\]
Since $p/q>1$ is fixed, the resulting lower bound is
$\Omega_{p,q}(n^2/\log n)$.
\end{proof}

For magic-state tensor powers in this exponentially small-overlap regime,
\cite[Corollary~3]{KS26} gives an $\Omega(n/\log n)$ conclusion. The
improvement in \cref{cor:tiny-overlap} comes from a uniform optimal-norm
bound, not from controlling the particular approximating coefficients.

\begin{corollary}[The quantitative rank--fidelity conjecture]\label{cor:MT-conj}
There is a universal constant $C$ such that every normalized state of
stabilizer rank at most $k$ satisfies $\FS(\ket{\psi})\ge2^{-k}/C$.
In fact,
\[
 \FS(\ket{\psi})\ge2^{-O(\sqrt{k\log(k+1)})}.
\]
\end{corollary}
\begin{proof}
By \cref{thm:span} with $\delta=0$ and monotonicity of $\Snd$,
\[
 \FS(\ket\psi)\ge\frac1{\Snd(k)}
 =2^{-12\sqrt{k\log_2(2k)}}.
\]
This proves the stronger bound. Since $12\sqrt{k\log_2(2k)}$ grows
more slowly than $k$, it is at most $k$ for all sufficiently large $k$.
For these $k$ the fidelity is at least $2^{-k}$. A single constant
$C\ge1$ covers the finitely many smaller values as well.
\end{proof}
This proves the quantitative conjecture following Theorem~1.3 of
\cite{MT25}. The robust form \eqref{eq:robust-fidelity} also allows every
fixed approximation error below one, independently of $k$.

\section{Rank lower bounds for pseudorandom states and unitaries}\label{sec:pseudo}

We use the Bell-difference fidelity test stated in \cref{prop:distinguisher}
as a black box. Our contribution is the rank threshold obtained from the
rank--fidelity inequality.

\subsection{A Bell-difference fidelity test}

The test detects stabilizer fidelity using copies of the unknown state,
without requiring a stabilizer decomposition. Our rank--fidelity bound
ensures that states of low approximate rank meet its acceptance threshold.
The pointwise rejection guarantee outside a small exceptional Haar set
allows us to amplify the test when such states occur only with
inverse-polynomial probability in an ensemble.

\begin{proposition}[Bell-difference fidelity test {\cite[Theorem~23 and Lemma~15]{GIKL23}}]\label{prop:distinguisher}
Let $\gamma=\gamma(n)\ge1$ be polynomially bounded in $n$.
For sufficiently large $n$, there is a test using $\poly(\gamma)$ copies
and $\poly(n,\gamma)$ time that accepts every pure state with stabilizer
fidelity at least $1/\gamma$ with probability at least $2/3$.
Outside an exceptional set of pure states of Haar measure at most
$\exp(-2^{\Omega(n)})$, its acceptance probability is at most $1/3$.
\end{proposition}

\subsection{Approximate stabilizer rank of pseudorandom states}
Constant-rank states were excluded from pseudorandom ensembles by
\cite{MT25}, and \cite{KS26} obtained a rank threshold of order
$\log n/\log\log n$. The following theorem raises this threshold to order
$(\log n)^2/\log\log n$ and gives a negligible-exception statement for
every fixed approximation error below one, using \cref{prop:distinguisher}.
\begin{definition}[Pseudorandom pure-state ensemble]\label{def:prs}
An efficiently preparable ensemble $\{\ket{\psi_{n,s}}\}_s$ on $n$ qubits is
pseudorandom if no polynomial-time quantum algorithm, given polynomially
many copies of one sampled state, distinguishes it with nonnegligible
advantage from the same number of copies of one Haar-random state.
A function is negligible if it is eventually smaller than $n^{-c}$ for
every $c>0$.
\end{definition}
This is the copy-access pseudorandomness notion of \cite{JLS18}; our
lower bounds only use its indistinguishability requirement.

\begin{theorem}[Negligible mass of low approximate rank]\label{thm:prs}
For every fixed $C>0$ and $0\le\delta<1$, a pseudorandom pure-state ensemble
satisfies
\begin{equation}\label{eq:prs}
 \Pr_s\left[
 \chi_\delta(\ket{\psi_{n,s}})
 \le C\frac{(\log_2n)^2}{\log_2\log_2n}
 \right]=\negl(n).
\end{equation}
The expression is used for sufficiently large $n$.
\end{theorem}
\begin{proof}
The idea is to use low approximate rank only as a guarantee of detectable
stabilizer fidelity. The test in \cref{prop:distinguisher} does not need to find a low-rank
decomposition.

Let $k_n$ be the integer part of the threshold in \eqref{eq:prs}, and
let $p_n$ be the probability of the low-rank event. Since
$\sqrt{k_n\log(2k_n)}=O_C(\log n)$, \cref{thm:span} gives every
state in this event stabilizer fidelity at least
\[
 f_0(n)=\frac{1-\delta^2}{\Snd(k_n)}\ge n^{-O_{C,\delta}(1)}.
\]
Apply \cref{prop:distinguisher} with $\gamma=1/f_0(n)$.
Repeating on fresh copies and taking a majority vote reduces the error
to $\beta$ outside the proposition's exceptional Haar set, with an
$O(\log(1/\beta))$ overhead.
It then accepts each state of fidelity at least $f_0(n)$ with probability
at least $1-\beta$, while its average Haar acceptance is at most
\begin{equation}\label{eq:haar-soundness}
 \beta+\exp\!\left(-2^{\Omega(n)}\right).
\end{equation}
The costs remain polynomial whenever $\log(1/\beta)=O(\log n)$.

Suppose, for a contradiction, that $p_n$ is not negligible. By the
definition of negligible, there is a fixed integer $b\ge1$ such that
$p_n\ge n^{-b}$ for infinitely many $n$. Run the test with
$\beta=n^{-b-2}$. On the given ensemble, it accepts with probability at
least $p_n(1-\beta)$: it accepts each state in the low-rank event with
probability at least $1-\beta$, and acceptance outside the event can
only increase the total. On a Haar-random state its average acceptance
is at most $\beta+\exp(-2^{\Omega(n)})$. The difference is therefore at least
\[
 p_n(1-\beta)-\beta-\exp(-2^{\Omega(n)})
 \ge\tfrac12 n^{-b}
\]
along that infinite subsequence, for sufficiently large $n$. This is a
nonnegligible advantage for a polynomial-time test, contradicting
pseudorandomness. Hence $p_n$ is negligible.
\end{proof}

\subsection{Approximate Clifford rank of pseudorandom unitaries}
\begin{definition}[Pseudorandom unitary ensemble]
An efficiently implementable unitary ensemble on $n$ qubits is pseudorandom
against forward queries if no polynomial-time quantum algorithm with
polynomially many oracle uses of the sampled unitary distinguishes it from
a Haar-random unitary with nonnegligible advantage. A stronger definition
allowing inverse queries also implies this requirement.
\end{definition}

\begin{theorem}[Frobenius rank exclusion]\label{thm:pru}
For every fixed $C>0$ and $0\le\delta<1$, a pseudorandom unitary ensemble
has negligible probability of
\begin{equation}\label{eq:pru}
 \rC{}_{,\delta\sqrt N}^{F}(U)
 \le C\frac{(\log_2n)^2}{\log_2\log_2n}.
\end{equation}
The approximating operator is not required to be unitary. The distinguisher
uses only forward applications of $U$.
\end{theorem}
\begin{proof}
We first show that a good Frobenius approximation to $U$ gives good
state-vector approximations to a constant fraction of its columns.
We then test one randomly chosen column using only forward queries.

Let $k_n$ be the integer part of
$C(\log_2 n)^2/\log_2\log_2 n$, the rank threshold in \eqref{eq:pru}.
Suppose $A$ has Clifford rank at most $k_n$ and
$\norm{U-A}_F\le\delta\sqrt N$.
For a uniform computational-basis input $x$,
\[
 \E_x\norm{(U-A)\ket x}^2=\frac{\norm{U-A}_F^2}{N}\le\delta^2.
\]
Put $\delta_*^2=(1+\delta^2)/2<1$. Markov's inequality gives
\[
 \Pr_x[\norm{(U-A)\ket x}\le\delta_*]
 \ge\frac{1-\delta^2}{1+\delta^2}=:c_\delta>0.
\]
To see why these columns have low approximate stabilizer rank, write
$A=\sum_{j=1}^{k_n}c_jC_j$, allowing zero coefficients if fewer terms
are needed. Then
\[
 A\ket x=\sum_{j=1}^{k_n}c_jC_j\ket x
\]
is a combination of at most $k_n$ stabilizer states, because each
Clifford $C_j$ sends the computational-basis stabilizer $\ket x$ to a
stabilizer. For at least a $c_\delta$ fraction of inputs, this vector
approximates the normalized state $U\ket x$ within error $\delta_*$.
Consequently \cref{thm:span} gives those columns stabilizer fidelity at
least
\[
 f_0(n)=\frac{1-\delta_*^2}{\Snd(k_n)}
 \ge n^{-O_{C,\delta}(1)}.
\]

The distinguishing algorithm chooses one uniform $x$, prepares repeated
copies of $U\ket x$, and applies \cref{prop:distinguisher} with
$\gamma=1/f_0(n)$, amplified to error $\beta$ as in the proof of
\cref{thm:prs}. The same $x$ is used for every copy in a run. If $U$ has the
claimed low-rank approximation, the acceptance probability, averaged
over $x$, is at least $c_\delta(1-\beta)$. The algorithm does not need
to know the approximating operator $A$ or which inputs are good.
For Haar-random $U$ and any fixed $x$, the state $U\ket x$ is Haar-random,
so the Haar acceptance is bounded by \eqref{eq:haar-soundness}.

Let $p_n$ be the ensemble probability of the event in \eqref{eq:pru}.
If it were nonnegligible, some fixed $b\ge1$ would give
$p_n\ge n^{-b}$ for infinitely many $n$. Choose $\beta=n^{-b-2}$.
The difference in acceptance probabilities would then be at least
\[
 p_nc_\delta(1-\beta)-\beta-\exp(-2^{\Omega(n)})
 \ge\tfrac12 c_\delta n^{-b}
\]
for sufficiently large $n$ in that subsequence. The test uses
polynomially many forward queries, so this contradicts pseudorandomness.
No assumption about independence of different columns is needed.
\end{proof}

\section{Tomography and Hamiltonian learning from rank bounds}\label{sec:tomography}

Throughout this section, $0<\eps<1$ is the accuracy parameter and
$0<\beta<1/2$ is the failure probability.

All learning consequences in this section are instantiations of the
bounded-extent learning results of \cite{ADGG26}, including the unitary
learning consequence described there using the framework of \cite{DJJMP26}.
We use these algorithms as black boxes, substituting our rank--norm bounds
into their complexity guarantees. In particular, our stronger stabilizer
extent bound replaces the bound of \cite{KS26} in the state-learning
application; our Clifford coefficient-norm bound gives the corresponding
unitary and Hamiltonian learning guarantees. The decompositions whose existence was proved in
\cref{sec:state-main,sec:operator} need not be known to the learner.

\subsection{Copy-only state tomography}
We use the polynomial complexity guarantee of \cite[Corollary~8.11]{ADGG26},
with the confidence and classical-output guarantees of
\cite[Theorem~4.2]{AD26}.
\begin{theorem}[Bounded-extent state tomography {\cite{AD26,ADGG26}}]\label{thm:AD-import}
Let $R\ge1$. Given copies of an unknown
normalized $n$-qubit state $\ket{\psi}$ with $\LS(\ket{\psi})\le R$, there is an
algorithm returning a classical description of a normalized stabilizer
combination $\ket{\widehat\psi}$ such that
$\Dtr(\ket{\psi},\ket{\widehat\psi})\le\eps$ with probability at least $1-\beta$.
Its sample and quantum time complexity are
\begin{equation}\label{eq:AD-import}
 \poly(n,R,1/\eps,\log(1/\beta)).
\end{equation}
\end{theorem}

\begin{theorem}[Tomography of low-stabilizer-rank states]\label{thm:state-tomography}
Given copies of an unknown normalized $n$-qubit pure state $\ket{\psi}$ with
$\chi(\ket{\psi})\le k$, a classical description of a pure state within trace
distance $\eps$ can be produced with failure probability at most $\beta$
and sample and quantum time complexity
\begin{equation}\label{eq:state-tomography}
 \poly(n,1/\eps,\log(1/\beta))\,
 2^{O(\sqrt{k\log(2k)})}.
\end{equation}
The output is a stabilizer combination, not necessarily of rank $k$.
\end{theorem}
\begin{proof}
\Cref{cor:state-envelope} gives
$\LS(\ket{\psi})\le\sqrt{\Snd(k)}$.
Use $R=\sqrt{\Snd(k)}$ in \cref{thm:AD-import}. Since
$\log R=O(\sqrt{k\log(2k)})$, polynomial dependence on $R$ gives
\eqref{eq:state-tomography}.
\end{proof}

At fixed accuracy and inverse-polynomial failure probability, the
algorithm is polynomial time when $k=O(\log^2n/\log\log n)$.

\subsection{Clifford-rank unitary tomography}\label{sec:clifford-tomography}
We use the combination of polynomial-time Clifford self-correction with
the learning framework of \cite{DJJMP26}, described in
\cite[Section~8.2.5]{ADGG26}.
\begin{theorem}[Bounded-Clifford-norm tomography {\cite{DJJMP26,ADGG26}}]\label{thm:unitary-import}
Let $U$ be an unknown $n$-qubit unitary with $\LC(U)\le R$, $R\ge1$.
With query access to $U$ as in \cite[Section~8.2.5]{ADGG26},
there is an algorithm which, with probability at least $1-\beta$, returns
a classical Clifford combination $\widehat U$ satisfying
\[
 \min_\theta\norm{\widehat U-e^{i\theta}U}_F\le\eps\sqrt N,
 \qquad
 \rC(\widehat U)=\poly(R/\eps).
\]
Its query and quantum time complexity are
\[
 \poly(n,R,1/\eps,\log(1/\beta)).
\]
The output operator need not be unitary.
\end{theorem}

\begin{corollary}[Tomography from Clifford rank]\label{cor:unitary-tomo}
For $\rC(U)\le k$, the algorithm in \cref{thm:unitary-import} has query
and quantum time complexity
\begin{equation}\label{eq:unitary-tomo}
 \poly(n,1/\eps,\log(1/\beta))\,
 2^{O(\sqrt{k\log(2k)})},
\end{equation}
and may output a Clifford combination of rank
$\poly(1/\eps)2^{O(\sqrt{k\log(2k)})}$.
\end{corollary}
\begin{proof}
A unitary has $\norm U_F=\sqrt N$.
\Cref{cor:main-envelope} gives $\LC(U)\le\sqrt{\Bnd(k)}$.
Substitute this value of $R$ in \cref{thm:unitary-import}.
\end{proof}

\subsection{Hamiltonians with bounded Clifford rank}\label{sec:hamilonian-tomography}
The Clifford coefficient-norm bound also gives a learning algorithm for
Hamiltonians of bounded Clifford rank.

\begin{theorem}[Bounded-Clifford-norm Hamiltonian learning {\cite{ADGG26}}]
\label{thm:H-import}
Suppose $H=H^\dagger$ is an $n$-qubit Hamiltonian with $\tr H=0$ and
$\LC(H)\le R$, $R\ge1$. With query access to the short-time evolutions
$e^{-itH}$ used in \cite[Corollary~8.15]{ADGG26}, an algorithm outputs
$\widehat H$ with Frobenius error at most $\eps\sqrt N$ and failure
probability at most $\beta$. Its query complexity, quantum time, and total
evolution time are bounded by
\[
 \poly(n,R,1/\eps,\log(1/\beta)).
\]
The algorithm uses evolution times $t=\poly(\eps/R)$, and the output has
Clifford rank $\poly(R/\eps)$.
\end{theorem}

\begin{corollary}[Rank-promised Hamiltonian learning]\label{cor:H-tomo}
Assume $\tr H=0$, $\rC(H)\le k$, and a known bound
$\norm H_{\mathrm{op}}\le b$ with $b\ge1$. Under the access assumptions
of \cref{thm:H-import}, the complexities are at most
\begin{equation}\label{eq:H-tomo}
 \poly(n,b,1/\eps,\log(1/\beta))\,
 2^{O(\sqrt{k\log(2k)})}.
\end{equation}
\end{corollary}
\begin{proof}
Since $\norm H_F\le b\sqrt N$,
$\LC(H)\le b\sqrt{\Bnd(k)}$ by \cref{cor:main-envelope}.
Use $R=b\sqrt{\Bnd(k)}$ in \cref{thm:H-import}.
A known smaller norm bound can always be enlarged to one.
\end{proof}

\section{Conclusion}\label{sec:discussion}

Randomized Pauli projections and clustering of stabilizer codespaces give
subexponential bounds on stabilizer extent and Clifford coefficient norm
in terms of rank. These bounds connect short decompositions to small
coefficient norms, yielding nearly quadratic rank lower bounds and
stronger guarantees for distinguishing and learning low-rank states and
operators.

A natural question is whether the exponent $O(\sqrt{k\log(k+1)})$ can be
improved to $O(\sqrt k)$. This would remove the logarithmic loss in our
AND and magic-state rank lower bounds. In the present proof, that loss
comes from the dimension cost of merging codespaces; a stronger
clustering bound could therefore lead to an improvement. More broadly,
is stabilizer extent polynomially bounded in stabilizer rank, and does
an analogous bound hold for the squared Clifford coefficient norm
relative to normalized Frobenius norm? Such bounds would give exponential
rank lower bounds for magic-state tensor powers and for quadratic-phase
representations of AND, respectively.

\bibliographystyle{alpha}
\bibliography{references}

\appendix
\section{Stabilizer normal forms and random bases}\label{sec:app-stabilizer}

\subsection{Affine supports and full-support decompositions}\label{sec:proof-affine-phase}
We record the normal-form facts used in the state-to-operator reduction.

\begin{lemma}[Affine phase normal form]\label{lem:affine-phase}
Let $\ket s$ be a normalized stabilizer state on $m$ qubits whose computational
support has codimension $t$.
\begin{enumerate}[label=(\roman*)]
\item Its support is an affine space $\mathcal A=\{x:Mx=b\}$ of
codimension $t$. It can be written, up to a global phase, as
\begin{equation}\label{eq:affine-form}
 \ket s=2^{t/2}\Pi_{\mathcal A}D\ket{+^m},
\end{equation}
where $D$ is a diagonal Clifford on all $m$ qubits and
\[
 \Pi_{\mathcal A}=\sum_{x\in\mathcal A}\ket x\bra x
\]
is the projector onto the span of computational-basis vectors labelled
by elements of $\mathcal A$.
\item It is a linear combination of $2^t$ full-support stabilizer states
of the form $D'\ket{+^m}$ with diagonal Clifford $D'$.
\item Moreover,
\[
 2^t=|G_s\cap\Zgrp_m|.
\]
\end{enumerate}

\end{lemma}
\begin{proof}
We prove the three claims in order, starting from the standard amplitude
normal form for stabilizer states \cite{DDM03}.

For (i), that normal form says that the nonzero amplitudes have equal
magnitude and lie on an affine space $\mathcal A$ of dimension $m-t$.
Here an affine space is the solution set of independent binary linear
equations $Mx=b$, and $t$ is the number of those equations. Choose
invertible affine coordinates $(y,z)$ so that $\mathcal A$ is given by
$z=0$, with $y\in\F_2^{m-t}$ free. On this support the phase has the form
\[
 i^{p(y)},\qquad
 p(y)=c+\sum_j d_jy_j+2\sum_{j<\ell}e_{j\ell}y_jy_\ell\pmod4,
\]
where $c,d_j\in\{0,1,2,3\}$ and $e_{j\ell}\in\{0,1\}$. This is what
is meant by a Clifford quadratic phase with values modulo four:
the linear terms are implemented by powers of the phase gate $S$,
and the quadratic terms by controlled-$Z$ gates.

Apply this diagonal Clifford to the $y$ coordinates, ignoring $z$.
The affine coordinate change is implemented by bit flips, swaps, and
CNOT gates. Conjugating by this basis permutation gives a diagonal
Clifford $D$ in the original coordinates whose phases agree with those
of $\ket s$ on $\mathcal A$. Projecting $D\ket{+^m}$ onto
$\mathcal A$ leaves $2^{m-t}$ amplitudes of magnitude $2^{-m/2}$.
Multiplication by $2^{t/2}$ normalizes them, proving
\eqref{eq:affine-form} up to global phase.

For (ii), write the rows of $M$ as $u_1,\ldots,u_t$. The projector
enforcing the equation $u_j\cdot x=b_j$ is
$(I+(-1)^{b_j}Z^{u_j})/2$. Multiplying these commuting projectors
enforces all the equations at once. Expanding the product gives
\[
 \Pi_{\mathcal A}
 =\prod_{j=1}^t\frac{I+(-1)^{b_j}Z^{u_j}}2
 =2^{-t}\sum_{z\in\F_2^t}
                   (-1)^{z\cdot b}Z^{M^Tz}.
\]
Consequently
\begin{equation}\label{eq:full-support-expand}
 \ket s=2^{-t/2}\sum_{z\in\F_2^t}
                  (-1)^{z\cdot b}DZ^{M^Tz}\ket{+^m},
\end{equation}
Each $DZ^{M^Tz}$ is diagonal Clifford, so every state in this
$2^t$-term expansion is a full-support stabilizer state.

For (iii), a signed diagonal Pauli $(-1)^cZ^u$ fixes $\ket s$ exactly
when every point in its support satisfies $u\cdot x=c$. The equations
holding throughout $\mathcal A$ are precisely the binary linear
combinations of its $t$ independent defining equations. There are
$2^t$ such combinations, including the trivial equation. Hence $H_s$
has $2^t$ diagonal elements. Forgetting their signs identifies them
with $G_s\cap\Zgrp_m$, proving $|G_s\cap\Zgrp_m|=2^t$.
\end{proof}

\subsection{A random Clifford basis}
\begin{lemma}[A random Clifford basis]\label{lem:random-basis}
Let $G_0\le\PG_m$ be a fixed maximal commuting subgroup, and let $G$ be
uniformly random among the maximal commuting subgroups of $\PG_m$.
Then
\begin{equation}\label{eq:intersection-expectation}
 \E_G|G\cap G_0|
 =1+\frac{2^m-1}{2^m+1}<2.
\end{equation}
Equivalently, for any fixed normalized stabilizer state $\ket s$ and a uniformly
random Clifford $F$ modulo global phase,
$\E_F|G_{Fs}\cap\Zgrp_m|=2^{m+1}/(2^m+1)$.
For $m=0$, we use $\HH_0=\C$ and the trivial phase-free Pauli
group $\PG_0=\{e\}$; both expectations equal one.
\end{lemma}
\begin{proof}
Assume $m\ge1$. Clifford conjugation is transitive on nonidentity
phase-free Paulis, i.e, for any two nonidentity Pauli classes represented by $P_1,P_2$, there is a Clifford operator $C$ such that $CP_1C^\dagger=P_2$ (up to phase). It is furthermore transitive even on maximal commuting subgroups, i.e, for any maximal commuting subgroups $G_1,G_2\subseteq \overline{\mathcal P}_m$, we have a Clifford operator $C$ such that $CG_1C^{\dagger}=G_2$.

Each such subgroup has $2^m-1$ nonidentity elements. Double counting therefore gives, for
every fixed $P\ne I$ in $\overline{\mathcal P}_m$,
\[
 \Pr[P\in G]=\frac{2^m-1}{4^m-1}=\frac1{2^m+1}.
\]
Count the elements of $G\cap G_0$, including $e$, and use linearity of expectation to obtain
\eqref{eq:intersection-expectation}. The equivalent Clifford-basis
form follows because $G_{Fs}=F G_s F^\dagger$ is uniform over the maximal
commuting subgroups. For $m=0$, all groups are trivial.
\end{proof}

\section{The linear-overhead state-to-operator reduction}\label{sec:app-reduction}
The direct proof of \cref{cor:state-envelope} already bounds state extent.
The following reduction records a more general relation: any operator
rank--norm inequality also yields a state rank--extent inequality, with
at most linear overhead in rank.

\begin{proposition}[Simultaneous full-support reduction]\label{prop:full-support}
If $V\subseteq\HH_n$ is spanned by $k\ge1$ stabilizer states, then there
are a Clifford $F$ and diagonal Cliffords $D_1,\ldots,D_r$, $r\le2k-1$,
such that
\begin{equation}\label{eq:span-full-support}
 FV\subseteq\Span\{D_j\ket{+^n}:1\le j\le r\}.
\end{equation}
The same $F$ and list work for every vector of $V$.
\end{proposition}
\begin{proof}
Let $\ket{s_1},\ldots,\ket{s_k}$ be normalized stabilizer states spanning $V$.
Choose a uniformly random Clifford $F$ and let $t_i(F)$ be the
codimension of the computational support of $F\ket{s_i}$. By
(iii) of \cref{lem:affine-phase},
\[
 2^{t_i(F)}=|G_{Fs_i}\cap\Zgrp_n|.
\]

\Cref{lem:random-basis} and linearity of expectation give
\[
 \E_F\sum_{i=1}^k2^{t_i(F)}
 =k\frac{2^{n+1}}{2^n+1}<2k.
\]
The quantity being averaged is integer-valued, so some single Clifford
$F$ satisfies $\sum_i2^{t_i(F)}\le2k-1$.

Fix this $F$. For each $i$, part (ii) of \cref{lem:affine-phase} expresses
$F\ket{s_i}$ as a linear combination of $2^{t_i(F)}$ full-support states
$D\ket{+^n}$ with $D$ diagonal Clifford. Collect all these states and
remove repetitions. The resulting list has length at most
$\sum_i2^{t_i(F)}\le2k-1$ and spans every $F\ket{s_i}$, hence all of $FV$.
The choice of $F$ and the list depend only on the supplied stabilizers,
not on the vector of $V$ to be represented.
\end{proof}

\begin{theorem}[Operator bounds imply state bounds]\label{thm:state-reduction}
Suppose a nondecreasing function $B$ satisfies
$\LC(A)^2\le B(k)2^{-n}\norm A_F^2$ whenever $\rC(A)\le k$.
Then every pure state $\ket{v}$ of stabilizer rank at most $k$ satisfies
\begin{equation}\label{eq:state-from-operator}
 \xi(\ket{v})\le B(2k-1)\norm{\ket{v}}^2.
\end{equation}
In particular, the operator bound of \cref{cor:main-envelope} gives
$\xi(\ket{v})\le\Bnd(2k-1)\norm{\ket{v}}^2$.
\end{theorem}
\begin{proof}
Apply \cref{prop:full-support} to a stabilizer span containing $\ket{v}$, and set
\[
 A_v=\sqrt N\diag((F\ket{v})_x)_{x\in\F_2^n}.
\]
If $F\ket{v}=\sum_{j=1}^r a_jD_j\ket{+^n}$ is the representation supplied by
\eqref{eq:span-full-support}, comparison of diagonal entries gives
$A_v=\sum_{j=1}^r a_jD_j$. Hence $\rC(A_v)\le r\le2k-1$, and
$A_v\ket{+^n}=F\ket{v}$. Moreover,
\[
 \frac{\norm{A_v}_F^2}{N}
 =\frac1N\sum_{x\in\F_2^n}|\sqrt N(F\ket{v})_x|^2
 =\sum_x|(F\ket{v})_x|^2=\norm{\ket{v}}^2.
\]
Conversely, any Clifford decomposition $A_v=\sum_jc_jC_j$ gives
$\ket{v}=\sum_jc_jF^\dagger C_j\ket{+^n}$, so $\LS(\ket{v})\le\LC(A_v)$.
Apply the hypothesis, using both the Clifford-rank bound and the exact
Frobenius identity.
\end{proof}

\begin{remark}[Why the diagonal lift is needed]
Writing $\ket{v}=\sum_{i=1}^k c_iC_i\ket{0^n}$ does give an operator
$A=\sum_i c_iC_i$ with $\rC(A)\le k$, $A\ket{0^n}=\ket{v}$, and
$\LS(\ket{v})\le\LC(A)$. What is missing is control of
$\norm A_F$ by $\sqrt N\norm{\ket{v}}$: cancellation on one input need
not control the other columns. For example, for $0<\varepsilon<1$,
\[
 A_\varepsilon=\varepsilon^{-1}I
       -(1-\varepsilon)\varepsilon^{-1}Z
 =\diag\!\left(1,\frac{2-\varepsilon}{\varepsilon}\right)
\]
satisfies $A_\varepsilon\ket0=\ket0$, but
$\norm{A_\varepsilon}_F\to\infty$ as
$\varepsilon\to0$.

The common basis change and diagonal lift achieve low Clifford rank and
$\norm{A_v}_F=\sqrt N\norm{\ket{v}}$ simultaneously. The ordinary Choi
comparison alone does not provide this controlled lift. The operator
inequality must apply to nonunitary operators, since $A_v$ is generally
not unitary.
\end{remark}

\section{Technical details}\label{sec:technical-details}

\subsection{Uniform rank--norm bounds}\label{sec:app-uniform-bounds}
\begin{stateboundrestated}[Uniform stabilizer rank--extent bound]

Define $\Snd(0)=0$ and, for integers $k\ge1$, set
\[
 \Snd(k)=2^{12\sqrt{k\log_2(2k)}}.
\]
Then $\chi(\ket{v})\le k$ implies $\xi(\ket{v})\le\Snd(k)\norm{\ket{v}}^2$.
\end{stateboundrestated}
\begin{proof}
The case $k=0$ is immediate. For $k\ge1$, put $L=\log_2(2k)$.
If $k\le8L$, the single-state bound has base-two
logarithm less than $k\le12\sqrt{kL}$. Otherwise take
$r=\ceil{\sqrt{k/(2L)}}$. The logarithm of the factor in
\eqref{eq:state-finite-bound} is at most
\[
 4rL+k/r+1
 \le3\sqrt2\sqrt{kL}+4L+1
 \le12\sqrt{kL},
\]
where the last inequality uses $k>8L$ and $L\ge1$.
\end{proof}

\begin{operatorboundrestated}[Uniform operator inequality]
Define $\Bnd(0)=0$ and, for integers $k\ge1$, set
\[
 \Bnd(k)=2^{12\sqrt{k\log_2(2k)}}.
\]
If $\rC(A)\le k$, then
\[
 \LC(A)^2\le\Bnd(k)2^{-n}\norm A_F^2.
\]
\end{operatorboundrestated}
\begin{proof}
The case $k=0$ is immediate. For $k\ge1$, let $L=\log_2(2k)$.
If $k\le8L$, \eqref{eq:single-bound} has base-two
logarithm less than $2k\le12\sqrt{kL}$. Otherwise take
$r=\ceil{\sqrt{k/(2L)}}$, which lies in $[2,k]$. The logarithm of the
factor in \eqref{eq:finite-bound} is at most
\[
 4rL+2k/r+2
 \le4\sqrt2\sqrt{kL}+4L+2
 \le12\sqrt{kL}.
\]
The last inequality follows from $k>8L$ (and $L\ge1$).
\end{proof}

\subsection{Logarithmic inversion}\label{sec:app-inversion}
\begin{inversionrestated}[Logarithmic inversion]
For $k,L\ge1$,
\[
 L^2=O\!\bigl(k\log(1+k)\bigr)
 \quad\Longrightarrow\quad
 k=\Omega\!\left(\frac{L^2}{\log(1+L)}\right).
\]
\end{inversionrestated}
\begin{proof}
Write $L^2\le Ck\log_2(1+k)$. If $k\ge L^2$, the conclusion is immediate;
otherwise $\log_2(1+k)\le\log_2(1+L^2)\le2\log_2(1+L)$, giving
$k\ge L^2/(2C\log_2(1+L))$.
\end{proof}

\end{document}